\documentclass[11pt]{article}
\usepackage[hmargin=1in,vmargin=1in]{geometry}
\usepackage{geometry}                    
\usepackage{amssymb}
\usepackage{amsthm}
\usepackage{amsmath}
\usepackage{color}
\usepackage{bbm}
\usepackage{graphicx}
\usepackage{url}
\usepackage[T1]{fontenc}

\newtheorem{proposition}{ Proposition}[section]

\newtheorem{remark}{ Remark}

\newtheorem{theorem}{ Theorem}[section]

\newtheorem{definition}{Definition}[section]
\newtheorem{example}{Example}[section]
\newtheorem{condition}{Condition}[section]

\newcommand{\indicator}[1]{\mathbf 1_{\left[ {#1} \right] }}

\newcommand{\E}{\mathbb E}
\newcommand{\Pspace}{\mathcal P}
\newcommand{\Hspace}{\mathbb H}
\newcommand{\Sspace}{\mathbb S}
\newcommand{\F}{{\mathcal F}}

\newcommand{\full}{{\mbox{\scriptsize{full}}}}

\newcommand{\mm}{\textup{d}}
\newcommand{\qq}{\textup{q}}
\newcommand{\dd}{\textup{d}}

\newcommand{\Fpartial}{{\mathcal F}^S }
\newcommand{\tr}{\top }
\newcommand{\premium}{\Pi}
\newcommand{\Xpi}{X^\pi}
\newcommand{\goodchi}{\protect\raisebox{2pt}{\large$\chi$}}

\newcommand{\w}{\text{w}}
\newcommand{\y}{\text{y}}

\begin{document}
\title{Backward SDEs for Control with Partial Information}
\author{A. Papanicolaou\thanks{Department of Finance and Risk Engineering, NYU Tandon School of Engineering, 6 MetroTech Center, Brooklyn NY 11201 {\em ap1345@nyu.edu}. Part of this research was performed while the author was visiting the Institute for Pure and Applied Mathematics (IPAM), which is supported by the National Science Foundation.}}
\maketitle
 \begin{abstract} This paper considers a non-Markov control problem arising in a financial market where asset returns depend on hidden factors. The problem is non-Markov because nonlinear filtering is required to make inference on these factors, and hence the associated dynamic program effectively takes the filtering distribution as one of its state variables. This is of significant difficulty because the filtering distribution is a stochastic probability measure of infinite dimension, and therefore the dynamic program has a state that cannot be differentiated in the traditional sense. This lack of differentiability means that the problem cannot be solved using a Hamilton-Jacobi-Bellman (HJB) equation. This paper will show how the problem can be analyzed and solved using backward stochastic differential equations (BSDEs), with a key tool being the problem's dual formulation.\\ 

 \vspace{.0005cm}
 
 \vspace{.0005cm}
\end{abstract}
\section{Introduction}
\label{sec:intro}
Consider an investor who seeks to optimally allocate among $(\dd+1)$-many assets: a risk-free instrument (e.g., a money-market or bank account) that pays interest rate $r\geq 0$, and $\dd$-many risky exchange-traded funds (ETFs) denoted $S=(S^1,S^2,\dots, S^\dd)^\tr$ where
\[S^i(t) = \hbox{time-$t$ price of the $i$\textsuperscript{th} ETF} \ .\]
These prices are continuous processes on a filtered probability space $\left(\Omega,(\F_t)_{t\leq T},\mathbb P\right)$. Let $W$ and $B$ denote a pair of $\F_t$ Brownian motions where $W\in C([0,T];\mathbb R^{\mm})$ and $B\in C([0,T];\mathbb R^{\qq})$ for a positive integer $\qq<\infty$, and with
\[dW(t) dW(t)^\tr=I_{\mm\times \mm}dt,\qquad dB(t)dB(t)^\tr=I_{\qq\times \qq}dt\qquad d W(t) dB(t)^\tr=0\ ,\]
where $I_{(\cdot)}$ denotes an identity matrix, and ${(\cdot)}^\tr$ denotes matrix/vector transpose. The ETFs' price process $S\in C([0,T];\mathbb R^\dd)$ has returns that depend on a stochastic factor $Y\in C([0,T];\mathbb R^{\qq})$, as given  by the following hidden Markov model,
\begin{eqnarray}
\label{eq:dS}
\frac{dS^i(t)}{S^i(t)}&=&h^i(Y(t))dt+ \sum_{j=1}^{\mm}\sigma_\w^{ij}dW^j(t)+ \sum_{j=1}^{\qq}\sigma_\y^{ij}dB^j(t)\hspace{.6cm}\hbox{(observed)\ ,}\\
\label{eq:dY}
dY(t)&=&b(Y(t))dt+a(Y(t))dB(t)\hspace{3.7cm}\hbox{(hidden)\ ,}
\end{eqnarray}
where the initial condition $Y(0)$ is unobserved and independent of $W$ and $B$. In order to ensure existence and uniqueness of strong solutions to the SDEs, the coefficients $a,b$, and $h$ are assumed to be $C^1$ and Lipschitz continuous, with matrix $a\in \mathbb R^{\qq\times \qq}$ satisfying the condition $\inf_{y\in\mathbb R^{\qq}}aa^\tr(y)>0$ (i.e., positive definiteness). The matrices $\sigma_\w$ and $\sigma_\y$ combine for the total covariance
\[\sigma=\Big(\sigma_\w\sigma_\w^\top+\sigma_\y\sigma_\y^\top\Big)^{1/2}\in\mathbb R^{\dd\times \mm}\ ,\]
where it is assumed there is a constant $\epsilon$ such that
\begin{equation}
\label{eq:sigmaBound}
0<\epsilon\leq \sigma\sigma^\tr \leq \frac1\epsilon<\infty\ ,
\end{equation}
i.e., $\sigma\sigma^\tr$ is positive definite and bounded.

Let $\Fpartial_t$ denote the $\sigma$-algebra generated by $\{S(u):u\leq t\}$ for any time $t\in[0,T]$. The investor must decide upon an $\Fpartial_t$-adapted allocation vector $\pi(t)\in\mathbb R^\dd$ where for each $i$ 
\[\pi^i(t)=\hbox{time-$t$ proportion of wealth in $i$\textsuperscript{th} ETF}\ .\]
Clearly $\Fpartial_t\subset\F_t$, and in particular $Y(t)$ is not observable given $\Fpartial_t$. Hence the investor will need to filter $Y(t)$ given $\Fpartial_t$, and then use this filter to make an optimal investment decision. For a given strategy $\pi$ the investor's wealth is the process $\Xpi\in C([0,T];\mathbb R^1)$ that is a semi-martingale with
\begin{align}
\label{eq:dXintro}
\frac{d\Xpi(t)}{\Xpi(t)} &=rdt+ \sum_{i=1}^\dd  \pi^i(t)\left(\frac{dS^i(t)}{S^i(t)}-rdt\right)\ ,
\end{align}
where $\pi$ is considered admissible if it is \textbf{S-integrable,} i.e., 
\[\sum_{i=1}^\dd  \int_0^T\left|\pi^i(t)\Xpi(t)\right|^2dt<\infty~~~~\hbox{almost surely,}\]
(see \cite{karatzas2007,KS99}). The investor has a concave utility function $U(x)$ and finds an optimal $\pi$ by solving for her optimal value function,
\[V(t,x)=\sup_\pi\E\left[U(X^\pi(T))\Big|\Fpartial_t\vee\{X^\pi(t)=x\}\right]\ ,\]
where the supremum is taken over all $\Fpartial_t$-adapted $\pi$'s. This is a non-Markov control problem because the optimal $\pi(t)$ will depend on the entire history $\Fpartial_t$. In particular, the filter for $Y(t)$ is a non-Markov process, and as the optimal control will depend in this filter it causes the entire problem to be non-Markov.

This paper analyzes this non-Markov problem with a specialized focus on the effects of partial information. As $\Fpartial_t\subset \F_t$, the investor with only $\Fpartial_t$ is said to be \textit{partially informed,} and naturally there is a disadvantage by not having the full information of $\F_t$. In particular, all processes in \eqref{eq:dS} and \eqref{eq:dY} would be observed if the investor had the information contained in $\F_t$, in which case it stands to reason that there would be an improvement from her optimal $\Fpartial_t$-adapted value function. An investor who observes the information in $\F_t$ is said to be \textit{fully informed.} 

The partially-informed investor will compute the posterior distribution of $Y(t)$ given $\Fpartial_t$, which she could use to write her optimal strategy in \textit{feedback form}\footnote{The definition of `feedback form' is given in \cite{reneSIAM}[Chapter 2] and in \cite{bjork}[Chapter 19].}, but such a characterization is a function of a probability measure, which means it is a function of an infinite-dimensional input. Functions with infinite-dimensional inputs are the main difficulty when solving the partial-information investment problem: the optimal control depends on a measure-valued state that cannot be differentiated in a traditional sense, and hence the problem cannot be solved using a Hamilton-Jacobi-Bellman (HJB) equation. This paper overcomes this difficulty by using backward stochastic differential equations (BSDEs). 

In solving the partial-information problem it is useful to recognize that the market is complete (see \cite{karatzasXue1991,sassHaussman2004}) and then to solve a dual problem. Indeed, under certain basic assumptions (see Condition \ref{cond:novikov}) partial information allows for asset prices to be written in complete-market form,
\[\frac{dS^i(t)}{S^i(t)}=\hat h^i(t) dt+ d\nu^i(t)\ ,\]
where $\hat h^i(t)=\E[h^i(Y(t))|\Fpartial_t]$, and $\nu^i(t)$ is the \textit{innovation} given by
\[\nu^i(t) = \int_0^t\left(\frac{dS^i(u)}{S^i(u)}-\hat h^i(u)du\right)\ ,\]
such that $\sigma^{-1}\nu(t)$ is a $\dd$-dimensional $\Fpartial_t$-adapted Brownian motion. Completeness of the market leads to considerable simplification, as there is a unique equivalent martingale measure (EMM) (i.e., an equivalent measure where $e^{-rt}S(t)$ is a local martingale), making the dual function a straight-forward conditional expectation (i.e., the dual problem's infimum over the set of EMM is trivial because the set is a singleton containing the unique EMM). As conditional expectations can be represented as solutions to BSDEs, it follows that the dual value function is the solution to a BSDE, from which the primal value function and optimal strategy can be computed as well.

In contrast to the partially-informed investor, the fully-informed does not need to filter because she observes the full $\F_t$, and therefore chooses an optimal $\F_t$-adapted $\pi$ that is obtained from a finite-state HJB equation and is written as a function of $\Xpi(t)$ and $Y(t)$. However, the full information model remains an incomplete market model because the $Y$ process cannot be bought or sold, thereby making it somewhat technical to solve the full-information HJB equation. Existence and regularity of solutions to this HJB equation can be shown when the SDE coefficients meet specific assumptions (see \cite{pham2002}). If they exist then HJB-based solutions are convenient, but it is still useful to solve the full-information problem using BSDEs because it allows for comparisons with the BSDEs from partial information. 

The investor's quantification of factor-latency is the so-called \textit{information premium}, or the expected loss in utility due to partial information. From the perspective of partial information, full information is an improvement in the sense that
\[\hspace{1cm}\E\left[V^\full(t,x,Y(t))\Big|\Fpartial_t\right]\geq V(t,x)\ ,\qquad\forall x\geq 0~\hbox{and}~\forall t\in[0,T]\ ,\]
where $V^\full(t,x,y)$ is the fully-informed investor's value function. This inequality is consistent with common-sense intuition that it is better to know the exact values of the $Y$ factors, but it is interesting to point out that this inequality shows how a complete-market investor can expect an improvement  if she were allowed to switch to an incomplete-market. It should also be pointed out that this is an expectation, and it may be possible for $V^\full(t,x,Y(t))<V(t,x)$ (see Example \ref{ex:linearExample}). Quantification of the information premium is an important question that is addressed in this paper using BSDEs.

\subsection{Literature Review}

Portfolio optimization builds on control theory and relies on concepts such as duality and concavity, which are presented in many books and papers including \cite{KS99} and \cite{rogers2002}. Initial works on consumption-portfolio choice and asset pricing under partial information include \cite{gennotte1986} who presents a separation theorem: agents first filter then optimize; \cite{detemple1986} with results on an economy with Gaussian information structures under partial information wherein a Kalman filter applies; \cite{basak2000,basak2005,detemple1994} with results on markets with multiple heterogeneous agents who update their beliefs with the arrival of financial innovations; \cite{dothanFeldman1986} shows how equilibrium interest rates with partial information have a trade-off between latent-variable persistence and the parameters controlling inference; and also \cite{feldman1989} which shows that the expectations hypothesis holds only if rates are non-stochastic. 

In \cite{karatzasXue1991} the partial information portfolio optimization problem is shown to reduce to a complete market problem, a result which is also shown in \cite{BDL2010}. Portfolio optimization with partial information and filtering is done in \cite{brendle2006,carmonaIP4,papanicolaouLee2016,wangWu2008}, but only for linear Gaussian cases. Greater generality and the role of martingales and duality theory are considered in \cite{lakner1998,pham2001}. There is also substantial literature dealing with partial information and (unobserved) regimes following finite-state Markov chains, such as \cite{bauerleRieder2005,sassHaussman2004}. The linear version of the full-information problem is addressed in \cite{kimOmberg}, with attention given to the so-called \textit{nirvana cases} where investors' expected utility is infinite. The role of forward-backward dynamics in portfolio optimization is shown in \cite{detemple1991} with a novel use of Malliavan calculus. Partial information with nonlinear filtering, BSDEs, and indifference pricing are considered in \cite{maniaSantacroce2010}, but under the assumption of a bound on $\sigma^{-1}h$, an assumption which is not made in this paper. 

Backward SDEs are covered in \cite{touzi2006,carmona2013,reneSIAM,elKarouiPengQuenez1997,kobylanski2000,pardouxRascanu,pham}, including important results for existence and uniqueness of solutions, and in \cite{elKarouiRouge2000} and \cite{imkeller2005} BSDEs are applied to problems of stochastic control for utility maximization. There is also an application of BSDEs in \cite{zhou2015} to robust utility maximization under volatility uncertainty. Path dependence and HJB equations with stochastic coefficients are considered in \cite{peng1992}, which can be compared to the BSDEs in non-Markov control problems. Another possibility is to write partial information's infinite-dimensional program using the master equation, similar to \cite{bensoussan2015}; the master equation uses G\^ateaux derivatives in an HJB-type equation with differentiation done over measure-valued inputs. Two important resources for control theory are \cite{flemingSonerBook}, and \cite{bensoussanBook} for control problems with partial information. There is also \cite{bensoussan2009} where the application of partial-information control methods are used to optimize in a market where the price on a basket of goods is noisy, and a modified Mutual Fund Theorem is obtained. Finally, a review of nonlinear filtering is found in \cite{bensoussanBook,florchingerLeGland1991} for the Zakai and Kushner-Stratonovich equations, and Monte-Carlo methods for approximation (i.e., the particle filter) are presented in \cite{cappeMoulinesRyden2005}.

\subsection{Main Results in this Paper}
This paper brings together results from filtering, duality, and BSDE theory, and uses them to solve the nonlinear partial-information optimal portfolio problem. This paper's application of BSDEs is significant because consideration is given to the case of unboundedness of the function $h(y)$ in \eqref{eq:dS}. If $h$ were bounded then the results from \cite{imkeller2005,maniaSantacroce2010} would apply. Unboundedness of $h$ is of considerable interest because it allows for extreme behavior among investors with low risk aversion, but it introduces some technical difficulty in proving existence and uniqueness of BSDE solutions; the proofs are provided in this paper and rely on some of the specific features of the partial-information finance problem. 

The BSDE approach is also used when comparing with full information and quantifying the information premium. The full-information problem is solved using BSDEs and the information premium is represented dynamically using the BSDE coefficients. The information premium is important because it gives quantitative evidence that information matters; partially-informed investors are at a disadvantage to the fully informed. 

The rest of the paper is organized as follows: Section \ref{sec:filteringAndControl} formalizes the filtering and control problem, and introduces the dual formulation; Section \ref{sec:BSDEs} shows how the problem can be solved using BSDEs when $U(x) $ is a power utility, with verification that the solution $\pi$ obtained from the BSDEs is in fact optimal --both for partial and full information; Section \ref{sec:nirvana} provides insight by considering the example of the Gaussian linear case; Section \ref{sec:nonlinearExample} gives a nonlinear example with a simulation of the BSDEs. Appendices \ref{app:partialInfo_BSDE}, \ref{app:fullInfoVerification}, and \ref{app:fullInfo_BSDE} contain technical proofs for the propositions and theorems of Section \ref{sec:BSDEs}.

\section{Filtering and Control}
\label{sec:filteringAndControl}
It will be assumed throughout that $h$ satisfies the Novikov condition

\begin{condition}[Novikov]
\label{cond:novikov}
The function $h$ is such that
\begin{equation}
\label{eq:novikov}
\E\exp\left(\frac{1}{2\epsilon}\int_0^T\left\|h(Y(t))-\mathbf r\right\|^2dt\right)<\infty\ ,
\end{equation}
where $\mathbf r=(r,r,\dots,r)^\tr\in\mathbb R^d$, $\|\cdot\|$ denotes the Euclidean norm, and $\epsilon>0$ is the bounding constant in \eqref{eq:sigmaBound}.
\end{condition}
Clearly \eqref{eq:novikov} holds for $h$ bounded, but it will be interesting to consider $h$ unbounded along with low risk aversion (these ideas will become clearer in later sections).

\subsection{Filtering}
\label{sec:filtering}
In matrix/vector form, the observations are given by 
\[\frac{dS(t)}{S(t)}= h(Y(t))dt+\sigma_\w dW(t)+\sigma_\y dB(t)\ .
\]
The filter is defined for an appropriate test function $g$ as
\[\hat  g(t) = \E\left[ g(Y(t))\Big|\Fpartial_t\right]\ ,\]
for any $g$ such that $\sup_{t\in[0,T]}\E\| g(Y(t))\|^2<\infty$. Using $\hat h(t)= \E[ h(Y(t))|\Fpartial_t]$, an important feature from filtering theory is the innovations process
\begin{equation}
\label{eq:dnu}
\nu(t) =\int_0^t \left(\frac{dS(u)}{S(u)}-\hat h(u)du\right)\ ,
\end{equation}
which is a Gaussian process, namely 
\[\zeta(t)=\sigma^{-1}\nu(t)\]
is $\Fpartial_t$-adapted $\dd$-dimensional Brownian motion. The innovations process is used to re-write equation \eqref{eq:dS} in a complete-market form,
\begin{align}
\label{eq:dScomplete}
\frac{dS(t)}{S(t)}&=\hat h(t)dt+\sigma d\zeta(t)\ .
\end{align}
This is a complete market because there is a unique equivalent martingale measure (EMM), namely $\frac{d\mathbb Q}{d\mathbb P}=Z(t)$ that is given by the Dolean-Dade exponent (due to Condition \ref{cond:novikov}),
\begin{align}
\nonumber
\frac{d\mathbb Q}{d\mathbb P}\Bigg|_{\Fpartial_t}&=Z(t) \\
\label{eq:Z}
&= \exp\left( -\frac12\int_0^t\left\|\sigma^{-1}(\hat h(u)-\mathbf{r} )\right\|^2 du-\int_0^t(\sigma^{-1}(\hat h(u)-\mathbf{r}) )^\tr d\zeta(u)\right) \ .
\end{align}

\subsection{Control for Optimal Terminal Wealth with Partial Information}
The investor chooses an $\Fpartial_t$-adapted strategy $(\pi(t))_{t\leq T}$ and has a self-financing wealth process, as given in equation \eqref{eq:dXintro}, that can be written using the innovations process,
\begin{align}
\nonumber
\frac{d\Xpi(t)}{\Xpi(t)} &=rdt+\sum_{i=1}^\dd\pi^i(t)(\hat h^i(t)-r)dt+ \sum_{i=1}^\dd \pi^i(t)d\nu^i(t)\ .
\end{align}
The investor's strategy is selected from an admissible set $\mathcal A$ given by
\begin{equation}
\label{eq:A}
 \mathcal A =\left\{\Fpartial_t\hbox{-adapted}~~\pi:[0,T]\times\Omega\rightarrow \mathbb R^\dd\ ,~~\hbox{s.t.}~~ \int_0^T\Big|\Xpi(t)\|\pi(t)\|\Big|^2dt<\infty~~\hbox{a.s.}\right\}\ ,
\end{equation}
(see \cite{karatzas2007,KS99}). For any $\pi\in \mathcal A$ the wealth process is almost surely non-negative, which rules out arbitrage from doubling strategies.

The investor has a utility function $U:\mathbb R^+\rightarrow\mathbb R^+$ that is concave and satisfies the Inada conditions:
\begin{condition}
\label{cond:inada}
The utility function $U(x)$ is continuously differentiable with $U'(x)>0$ and $U''(x)<0$ for all $x\geq 0$, and satisfies the Inada conditions, $\lim_{x\nearrow\infty}U'(x)=0$ and $\lim_{x\searrow0}U'(x)=\infty$.
\end{condition}
The utility function used throughout this paper is of constant relative risk aversion (CRRA), or simply the power utility, 
\[U(x) = \frac{1}{1-\gamma}x^{1-\gamma},\]
for $\gamma> 0$ and $\gamma\neq 1$. The investor seeks to maximize expected terminal utility of discounted wealth, with her control being selected from the class of admissible strategies given in \eqref{eq:A}. This leads the investor to find her optimal value function $V(t,x)$, which is formally written as a supremum over strategies in $\mathcal A$,
\begin{equation}
\nonumber
V(t,x) =\sup_{\pi\in\mathcal A} \E\left[U\left(\Xpi(T)\right)\Big|\Fpartial_t\vee\{\Xpi(t)=x\}\right]
\end{equation}
for all $x>0$. 

The nonlinearity introduced by the supremum can be avoided by considering the dual formulation of this problem. Let $V^\star $ denote the solution to the dual value function (see \cite{lakner1998,rogers2002,KS99}), 
\begin{align}
\label{eq:dualV}
V^\star(t,p) &= \inf_{\mathbb Q\ll \mathbb P}\E\left[ U^\star\left(pe^{-r(T-t)}\frac{d\mathbb Q}{d\mathbb P}\Big|_{\Fpartial_T}\right)\Bigg|\Fpartial_t\right]=\E\left[U^\star\left(pe^{-r(T-t)}\frac{Z(T)}{Z(t)}\right)\Bigg|\Fpartial_t\right]\ ,
\end{align}
for all $p>0$ where $\mathbb Q\ll \mathbb P$ denotes the family of equivalent probability measures under which $e^{-rt}S(t)$ is an $\Fpartial_t$ (local) martingale. Clearly, completeness of the market and the unique EMM given by equation \eqref{eq:Z} are the reason why the infimum is dropped in \eqref{eq:dualV}. The dual value function $V^\star$ is also a non-Markov process, yet it will be seen in Section \ref{sec:BSDEs} that it can be expressed using BSDEs, and hence it will be possible to obtain tractable representations of the solution to the partial-information control problem.

For continuous processes driven by Brownian motions, the general relationship between $V$ and $V^\star$ is discussed in \cite{rogers2002}, namely that
\begin{align*}
V^\star(t,p)&=\sup_{x>0}\left(V(t,x)-xp\right)\qquad\hbox{for all $p>0$}\ .
\end{align*}
In general $V(t,x)\leq \inf_p\left(V^\star(t,p)+xp\right)$ for all $x>0$, but if Condition \ref{cond:inada} holds and if $V(t,x)<\infty$ for some $x>0$, then $V$ and $V^\star$ are conjugates (i.e., they are Fenchel-Legendre transforms of one another),\footnote{Inada conditions and concavity are the main requirements for conjugacy in a complete market. In comparison, conjugacy in an incomplete market requires the additional condition of \textit{asymptotic elasticity,} $\overline\lim_{x\rightarrow\infty}xU'(x)/U(x)<1$ as shown in \cite{KS99}.} 
\begin{align*}
V(t,x)&=\inf_{p>0}\left(V^\star(t,p)+xp\right)\qquad\hbox{for all $x>0$}\ .
\end{align*}
For the power utility, first-order conditions yield the transform 
\[U^\star(p) = \frac{\gamma }{1-\gamma}p^{-\frac{1-\gamma}{\gamma}}\ , \]
and the expression in \eqref{eq:dualV} can be rewritten as
\begin{equation}
\label{eq:Vstar_rep}
V^\star(t,p)=U^\star\left(pe^{-r(T-t)}\right)\xi(t) =  \frac{\gamma }{1-\gamma}\left(pe^{-r(T-t)}\right)^{-\frac{1-\gamma}{\gamma}} \xi(t)\ ,
\end{equation}
where 
\[\xi (t) = Z(t)^{ \frac{1-\gamma}{\gamma}}\E\left[Z(T)^{-\frac{1-\gamma}{\gamma}}\Big|\Fpartial_t\right] \ .\]
For $p>0$ it is clear that $V^\star$ is a finite and strictly convex function if $|\xi(t)|<\infty $ almost surely. To ensure finiteness of $\xi(t)$, the model parameters and the risk aversion must permit the following condition:

\begin{condition}
\label{cond:mgf_Z}
The model parameters in \eqref{eq:dS}, \eqref{eq:dY}, and the power utility's risk aversion $\gamma$, are such that
\begin{align*}
\E\exp\left(\frac{2|\gamma-1||\gamma-2|}{\epsilon\gamma^2}\int_0^T\|\hat h(t)\|^2dt\right)<\infty\ ,
\end{align*}
where $\epsilon>0$ is the bounding constant given in \eqref{eq:sigmaBound}, with the derivation of this bound following from Proposition \ref{prop:sup_xiK_bound}. This bound ensures $|V^\star(t,p)|<\infty$ for all $p\in(0,\infty)$. 
\end{condition}
The set of $h$ functions for which Condition \ref{cond:mgf_Z} holds is not empty, as shown in the following remark.
\begin{remark}[Nonlinear $h$ Satisfying Condition \ref{cond:mgf_Z}]
\label{remark:nonEmptyH}
Partial information can be reduced to a condition on full information through multiple applications of Jensen's inequality,
\begin{align*}
&\mathbb E\exp\left(\frac{2|\gamma-1||\gamma-2|}{\epsilon\gamma^2}\int_0^T\|\hat h(t)\|^2dt\right)\\
&\leq \frac1T\int_0^T\mathbb E\exp\left(\frac{2T|\gamma-1||\gamma-2|}{\epsilon\gamma^2}\| h(Y(t))\|^2\right)dt\ .
\end{align*}
Hence, Condition 2.3 is satisfied for any $h$ such that 
\[\sup_{t\in[0,T]}\mathbb E\exp\left(\frac{2T|\gamma-1||\gamma-2|}{\epsilon\gamma^2}\| h(Y(t))\|^2\right)dt<\infty\ .\] Certainly this includes bounded nonlinear functions. An explicit example in one dimension involves $Y$ being a Cox-Ingersoll-Ross (CIR) process,
\[dY(t) = \kappa(\bar Y-Y(t))dt + a\sqrt{Y(t)}dB(t)\]
where $\kappa>0$, $\bar Y>0$, $0<a^2\leq 2\bar Y\kappa$, and $h(y) = \sqrt  y$, with a sufficient condition for Condition \ref{cond:mgf_Z} being $\frac{2T|\gamma-1||\gamma-2|}{\epsilon\gamma^2}<\frac{2\kappa}{a^2}$. Note that this example does not have the condition of $\inf_{y\in\mathbb R^{\qq}}aa^\tr(y)>0$, but this does not pose an issue because the SDE for $Y(t)$ is well defined for $a^2\leq 2\bar Y\kappa$. Section \ref{sec:nonlinearExample} will explore this example further.
\end{remark}

The need for Condition \ref{cond:mgf_Z} is seen in the proof of Proposition \ref{prop:sup_xiK_bound} in Appendix \ref{app:partialInfo_BSDE}, from which it is seen that
\begin{align}
\nonumber
\E\sup_{t\in[0,T]}|\xi(t)|^2&\leq\E \sup_{t\in[0,T]}\left(\frac{Z(T)}{Z(t)}\right)^{-2\frac{1-\gamma}{\gamma}}\\
\label{eq:finiteXi}
&\leq\E\exp\left(\frac{2|\gamma-1||\gamma-2|}{\epsilon\gamma^2}\int_0^T(\|\hat h(t)\|^2+\|\mathbf{r}\|^2)dt\right)<\infty\ ,
\end{align}
and it will be important in Section \ref{sec:BSDEs} to have $\E\sup_{t\in[0,T]}|\xi(t)|^2<\infty$ as part of the existence and uniqueness theory for $\xi$ to be a solution to a BSDE. Furthermore, defining $G(t)$ to be
\[G(t)= \xi(t)^{\gamma}\ ,\]
due to Condition \ref{cond:mgf_Z}, the finiteness in equation \eqref{eq:finiteXi} implies conjugacy of the Legendre transforms,
\begin{align}
\nonumber
V(t,x) &= \inf_{p>0}\left(V^{\star}(t,p)+xp\right)\\
\nonumber
&= \inf_{p>0}\left(\frac{\gamma }{1-\gamma}\left(pe^{-r(T-t)}\right)^{-\frac{1-\gamma}{\gamma}} \xi(t)+xp\right)\\
\label{eq:Vansatz}
&=U\left(xe^{r(T-t)}\right)G(t)\ .
\end{align}
Equation \eqref{eq:Vansatz} allows for optimal solutions to be obtained by solving the dual problem, with the optimal $V$ being obtained via straightforward (numerical) calculation of a Fenchel-Legendre transform on $V^\star$. Regardless of the chosen function to be computed, nonlinear filtering causes $V$ and $V^\star$ to require specially-designed backward recursive algorithms because of infinite dimensionality in the conditioning. To be more precise, the conditioning on $\Fpartial_t$ is an infinite-dimensional object and for numerical methods will need to be replaced with a finite-dimensional approximation. Sometimes there are ways to write the filtering distribution in a finite-dimensional form (e.g., using a Kalman filter \cite{brendle2006}, or finite-dimensional Markov chains \cite{bauerleRieder2005}), but general nonlinear filtering doesn't have such forms. 

Before starting the next section it is important to define the concept of \textit{investor nirvana.} Nirvana is defined in \cite{kimOmberg} as follows:

\begin{definition}[Investor Nirvana]
\label{def:nirvana}
For unbounded $U$, an investor achieves nirvana at $(t,x)$ if $V(t,x)=\infty$. For bounded $U$, nirvana is achieved when $V(t,x) = \max_xU(x)$.
\end{definition}
Nirvana can occur for a variety parameter regimes (see \cite{kimOmberg}), in particular power-utility investors with low risk aversion can achieve nirvana in the linear problem (see Section \ref{sec:nirvana}). Definition \ref{def:nirvana} will be used in Section \ref{sec:premium} when comparing the value functions of the partially-informed and the fully-informed investor. Definition \ref{def:nirvana} will also be used when considering cases where $V^\star(t,p)=\infty$ because it may be unclear whether there is nirvana or a duality gap (i.e., strict inequality such that $V(t,x)<\inf_{p>0}\left(V^\star(t,p)+xp\right)=\infty$ for some $x>0$). 
\begin{proposition}
\label{prop:nirvana}
In the partial-information case, investor nirvana cannot occur for $\gamma\in(0,1)$ if Condition \ref{cond:mgf_Z} holds, and cannot occur for $\gamma>1$ given \eqref{eq:novikov}.
\end{proposition}
\begin{proof} 
For $\gamma>1$ it follows from equations \eqref{eq:Vstar_rep} and \eqref{eq:Vansatz} that 
\[U\left(xe^{r(T-t)}\right)\leq V(t,x)\leq \inf_p\left(V^\star(t,p)+xp\right)=U\left(xe^{r(T-t)}\right) \xi(t)^\gamma\leq0\ ,\]
for all $x\in(0,\infty)$, implying that $0\leq \xi(t)\leq 1$. From equation \eqref{eq:novikov} it follows that 
\[\mathbb P\Big(\inf_{t\in[0,T]}\log(Z(T)/ Z(t))=-\infty\Big)=0\ ,\]
so that $\mathbb P\left((Z(T)/Z(t))^{-\frac{1-\gamma}{\gamma}}>0|\Fpartial_t\right)>0$ almost surely. This implies 
\[\xi(t)=\E\left[\left(\frac{Z(T)}{Z(t)}\right)^{-\frac{1-\gamma}{\gamma}}\Big|\Fpartial_t\right] >0\]
almost surely.

For $\gamma\in(0,1)$ with Condition \ref{cond:mgf_Z} not being violated, it follows that equation \eqref{eq:finiteXi} holds, and so $\xi(t)<\infty$ almost surely for all $t\in[0,T]$. Hence, $V(t,x)\leq \inf_p\left(V^\star(t,p)+xp\right)<\infty$ for all $x\in(0,\infty)$.
\end{proof}

\begin{remark}[Other Utility Functions]
\label{remark:otherUtilities}
This paper considers the problem only for power utility function. However, for exponential utility there should be results similar to power utility with $\gamma>1$, although there may be some technical difficulties in adapting the Inada conditions and wealth process to the entire real line. Log utility is a simple case that does not require BSDEs, as the optimal solution is simply the myopic strategy (see \cite{GKSW2014,lakner1998}). 
\end{remark}

\section{Solutions Using Backward Stochastic Differential Equations (BSDEs)}
\label{sec:BSDEs}
The partial-information dual function $V^\star(t,p)$ can be obtained by solving a BSDE. Solutions to BSDEs are constructed in the following function spaces,
\begin{align}
\nonumber
\Pspace_\dd&=\Big\{\hbox{the set of $\dd$-dimensional $\Fpartial_t$-adapted measurable processes on $\Omega\times[0,T]$}\Big\}\\
\nonumber
\Hspace_T^2(\Pspace_\dd)&=\left\{y\in \Pspace_\dd~~\hbox{s.t.}~~\E\int_0^T\|y(t)\|^2dt<\infty \right\}\\
\nonumber
\Sspace_T^2(\Pspace_\dd)&=\left\{y\in \Pspace_\dd\cap C([0,T];\mathbb R^\dd)~~\hbox{s.t.}~~\E\sup_{t\in[0,T]}\|y(t)\|^2<\infty \right\}\\
\label{eq:spaces}
\Sspace_T^\infty(\Pspace_\dd)&=\left\{y\in \Pspace_\dd\cap C([0,T];\mathbb R^\dd)~~\hbox{ s.t.}~~\sup_{t\in[0,T]}\|y(t)\|<\infty~~\hbox{a.s.}\right\}\ .
\end{align}
This section has the derivation of the BSDE for $V^\star(t,p)$ given by \eqref{eq:Vstar_rep}, and will give the conditions for existence and uniqueness. 

\subsection{The Partial-Information Value Function}
\label{sec:partialInfo_BSDE}
Define the martingale 
\[M(t)= \E\left[Z(T)^{-\frac{1-\gamma}{\gamma}}\Big|\Fpartial_t\right]\qquad\qquad\qquad\hbox{for }0\leq t\leq T\ ,\]
so that 
\[\xi(t) = Z(t)^{ \frac{1-\gamma}{\gamma}}M(t)\ .\]
Condition \ref{cond:mgf_Z} ensures $M(t)$ is square integrable, $\E M(t)^2<\E M(T)^2=\E Z(T)^{-2\frac{1-\gamma}{\gamma}}<\infty$, and allows for a unique representation of $M(t)$ as
\begin{align}
\label{eq:M-martingaleRep}
M(t) &= \E \left[Z(T)^{-\frac{1-\gamma}{\gamma}}\right]+ \sum_{i=1}^{\mm}\int_0^tM(u)\theta^i(u)d\zeta^i(u)\qquad\qquad\hbox{for }0\leq t\leq T\ ,
\end{align}
where $\theta(t)$ is the unique $\Fpartial_t$-adapted process with $\E\int_0^TM(u)^2\|\theta(u)\|^2du<\infty$ (see \cite{BDL2010}). In fact, $\theta$ is square-integrable by itself, $\theta\in \Hspace_T^2(\Pspace_\dd)$ (see Proposition \ref{prop:thetaIntegrability} in Appendix \ref{app:partialInfo_BSDE}). The representation in \eqref{eq:M-martingaleRep} should not be confused with the standard martingale representation theorem because the filtration generated by $\zeta$ may be smaller than $\Fpartial_t$.

Using the representation of \eqref{eq:M-martingaleRep}, the dual value function $V^\star(t,p) $ is given by the ansatz \eqref{eq:Vstar_rep} and the pair $\left( \xi,\alpha\right)\in \Sspace_T^2(\Pspace_1)\times\Hspace_T^2(\Pspace_\dd)$ that solves a BSDE.

\begin{theorem}
\label{thm:partialInfo_BSDE}
Assume Condition \ref{cond:mgf_Z}. The process $ \xi(t)$ in the representation of $V^\star$ in \eqref{eq:Vstar_rep} is given by the unique pair $\left( \xi,\alpha\right)\in\Sspace_T^2(\Pspace_1)\times\Hspace_T^2(\Pspace_\dd)$ that solves the BSDE,
\begin{align}
\label{eq:partialInfo_BSDE}
-d \xi(t) &=\beta(t,\alpha(t),\xi(t))dt-\sum_{i=1}^{\mm}\alpha^i(t)d\zeta^i(t)\ ,\\
\nonumber
 \xi(T)&=1\ ,
\end{align}
where 
\begin{align*}
\beta(t,\alpha(t),\xi(t))&=\frac{1-\gamma}{\gamma}\sum_{i=1}^{\mm}\left(\sigma^{-1}(\hat h(t)-\mathbf{r})\right)^i\alpha^i(t)+\frac12\frac{1-\gamma}{\gamma^2}\left\|\sigma^{-1}(\hat h(t)-\mathbf{r})\right\|^2\xi(t)\\ 
&=\frac{1-\gamma}{2}\left\|\sigma^{-1}\frac{\hat h(t)-\mathbf{r}}{\gamma}+\frac{\alpha(t)}{\xi(t)}\right\|^2\xi(t)-\frac{1-\gamma}{2|\xi(t)|}\|\alpha(t)\|^2\ .
\end{align*}
\end{theorem}
Some remarks are in order before starting the proof of Theorem \ref{thm:partialInfo_BSDE}.
\begin{remark}[Existence of Solutions to \eqref{eq:partialInfo_BSDE}]
\label{remark:partialInfoBSDE}
It should be noted that existence of a solution to \eqref{eq:partialInfo_BSDE} is due to Condition \ref{cond:mgf_Z}, as it allows for the martingale representation in \eqref{eq:M-martingaleRep}, from which a solution is constructed in terms of $\theta$ and $\hat h$,
\begin{align}
\label{eq:G}
\xi(t)&=M(0)\exp\left(-\int_0^t\left(\frac{\beta(u,\alpha(u),\xi(u))}{\xi(u)}+\frac12\left\|\frac{\alpha(u)}{\xi(u)}\right\|^2\right)du+\int_0^t\frac{\alpha(u)^\tr}{\xi(u)} d\zeta(u)\right)\\
\label{eq:alpha}
\frac{\alpha(t)}{\xi(t)}&=\theta(t)-\frac{1-\gamma}{\gamma}\sigma^{-1}\Big(\hat h(t)-\mathbf{r}\Big)\ ,
\end{align}
where it can be checked that $\xi(t)=Z(t)^{ \frac{1-\gamma}{\gamma}}M(t)$ and $\alpha(t)$ is the diffusion term from the It\^o differential of $d\left(Z(t)^{ \frac{1-\gamma}{\gamma}}M(t)\right)$, and hence it follows from Condition \ref{cond:mgf_Z} that $(\xi,\alpha)\in\Sspace_T^2(\Pspace_1)\times \Hspace_T^2(\Pspace_\dd)$. However, it should also be pointed out that $\theta$ is not easily obtained from the martingale representation theorem, but rather is found by solving the BSDE. On the other hand, BSDEs have explicit solution in very few cases, and so numerical methods should be used to find $(\xi,\alpha)$ and $\theta$.
\end{remark}

\begin{remark}[Uniqueness of Solutions to \eqref{eq:partialInfo_BSDE}]
Formulas \eqref{eq:G} and \eqref{eq:alpha} show the existence of a solution to equation \eqref{eq:partialInfo_BSDE} when Condition \ref{cond:mgf_Z} holds. If the function $h$ is bounded, then the coefficient $\beta$ is uniformly Lipschitz and uniqueness follows from an application of the existing theory (see \cite{reneSIAM,elKarouiPengQuenez1997,pham}). The proof for $h$ unbounded uses a truncation argument to show that solutions are a unique limit from a sequence of bounded problems (see Propositions \ref{prop:xiAnd_xim_paths} and \ref{prop:sup_xiK_bound}).
\end{remark}

\begin{remark}
Condition \ref{cond:mgf_Z} may be violated for $\gamma$ near zero, in which case formulas \eqref{eq:G} and \eqref{eq:alpha} do not provide a solution and there could be investor nirvana. It follows from \eqref{eq:G} that in terms of $\alpha(t)$, investor nirvana means 
\begin{align*}
&\mathbb P( \log\xi(t)=\infty)\\
&=\mathbb P\left( -\int_0^t\left(\frac{\beta(u,\alpha(u),\xi(u))}{\xi(u)}+\frac12\left\|\frac{\alpha(u)}{\xi(u)}\right\|^2\right)du+\int_0^t\frac{\alpha(u)^\tr}{\xi(u)} d\zeta(u)=\infty\right)>0\ ,
\end{align*}
 for some $t\in[0,T)$, which is certainly not the case for any $\theta\in \Hspace_T^2(\Pspace_\dd)$.
\end{remark}

\begin{proof}[Proof of Theorem \ref{thm:partialInfo_BSDE}]
The martingale representation in \eqref{eq:M-martingaleRep} is used to write a forward SDE

\begin{align*}
&d \xi(t) \\
&= M(t)d\left( Z(t)^{\frac{1-\gamma}{\gamma}}\right)+ Z(t)^{\frac{1-\gamma}{\gamma}}dM(t)+dM(t)\cdot d\left( Z(t)^{\frac{1-\gamma}{\gamma}}\right)\\
&=Z(t)^{\frac{1-\gamma}{\gamma}}M(t)\sum_{i=1}^{\mm} \left(-\frac{1-\gamma}{\gamma}\left(\sigma^{-1}(\hat h(t)-\mathbf{r})\right)^i+\theta^i(t)\right)d\zeta^i(t)\\
&\hspace{1cm}-\frac{1-\gamma}{\gamma}Z(t)^{\frac{1-\gamma}{\gamma}}M(t)\sum_{i=1}^{\mm}\left(\left(\sigma^{-1}(\hat h(t)-\mathbf{r})\right)^i\theta^i(t)\right)dt\\
&\hspace{2cm}+\frac12\frac{(1-\gamma)(1-2\gamma)}{\gamma^2}Z(t)^{\frac{1-\gamma}{\gamma}}M(t)\left\|\sigma^{-1}(\hat h(t)-\mathbf{r})\right\|^2dt\\
&=-\xi(t)\sum_{i=1}^{\mm}\underbrace{\left(\frac{1-\gamma}{\gamma}\left(\sigma^{-1}(\hat h(t)-\mathbf{r})\right)^i-\theta^i(t)\right)}_{=-\alpha^i(t)}d\zeta^i(t)\\
&-\underbrace{\frac{1-\gamma}{\gamma}\xi(t)\left(\sum_{i=1}^{\mm}\left(\sigma^{-1}(\hat h(t)-\mathbf{r})\right)^i\theta^i(t)-\frac12\left(\frac{1-2\gamma}{\gamma}\right)\left\|\sigma^{-1}(\hat h(t)-\mathbf{r})\right\|^2\right)}_{=\beta(t,\alpha(t),\xi(t))}dt\ ,
\end{align*}
which is \eqref{eq:partialInfo_BSDE} with $\alpha(t)$ and $\beta(t,\alpha,\xi)$ given accordingly. Equation \eqref{eq:partialInfo_BSDE} has non-Lipschitz coefficients if $h$ is not bounded, and therefore uniqueness of solutions is not covered by the general theory for solutions to BSDEs given in \cite{reneSIAM,elKarouiPengQuenez1997,pham}. Instead, uniqueness is shown using a truncation argument and the probabilistic representation of $\xi$ given in \eqref{eq:Vstar_rep}.

For some positive $K<\infty$, define the truncated filter,
\[\hat h_K(t)=\left\{
\begin{array}{cl}K\frac{\hat h(t)}{\|\hat h(t)\|}\ ,&\hbox{if }\|\hat h(t)\|\geq K\\
&\\
\hat h(t),&\hbox{otherwise,}
\end{array}\right.
\]
and consider the bounded BSDE
\begin{align}
\label{eq:partialInfo_boundedBSDE}
-d \xi_K(t) &=\beta_K(t,\alpha_K(t),\xi_K(t))dt-\sum_{i=1}^{\mm}\alpha_K^i(t)d\zeta^i(t)\ ,\\
\nonumber
\xi_K(T)&=1\ ,
\end{align}
where $\beta_K$ is the same drift function from \eqref{eq:partialInfo_BSDE} but with $\hat h_K(t)$ replacing the unbounded $\hat h(t)$. This drift parameter is linear with uniform linear growth bounds,
\begin{align*}
&\left|\beta_K(t,\alpha_K(t),\xi_K(t))\right|\\
&\leq \frac{|1-\gamma|}{\gamma}\left\|\sigma^{-1}(\hat h_K(t)-\mathbf{r})\right\|\|\alpha_K(t)\| +\frac{|1-\gamma|}{2\gamma^2}\left\|\sigma^{-1}(\hat h_K(t)-\mathbf{r})\right\|^2 |\xi_K(t)|\\
&\leq C_K\left(\|\alpha_K(t)\| + |\xi_K(t)|\right)\ ,
\end{align*}
which also serves as a uniform Lipschitz constant. Therefore, equation \eqref{eq:partialInfo_boundedBSDE} fits into the framework of \cite{reneSIAM,elKarouiPengQuenez1997,pham} and has solution $(\xi_K,\alpha_K)$ that is unique in the space $\Sspace_T^2(\Pspace_1)\times\Hspace_T^2(\Pspace_\dd)$.

Now define the stopping time $\tau_K=\inf\left\{t\geq 0:\|\hat h(t)\|\geq K\right\}$, and notice that $\tau_K\nearrow\infty$ almost-surely as $K\nearrow\infty$ because $\hat h(t)$ is integrable (due to the Novikov Condition in \eqref{eq:novikov}). Then using the fact that $\left(|\xi(t)-\xi_K(t)|\right)\indicator{\tau_K\geq T}=0$ from Proposition \ref{prop:xiAnd_xim_paths}, and also using the bound $\sup_{K>0}\E\sup_{t\in[0,T]}|\xi_K(t)|^2<\infty$ from Proposition \ref{prop:sup_xiK_bound}, it is shown that $\xi_K$ converges in mean,
\begin{align*}
\E\sup_{t\in[0,T]}|\xi(t)-\xi_K(t)|&=\E\sup_{t\in[0,T]}|\xi(t)-\xi_K(t)|\indicator{\tau_K<T}\\
&\leq \E\sup_{t\in[0,T]}\left(|\xi(t)|+|\xi_K(t)|\right)\indicator{\tau_K<T}\\
&\leq \left(\left(\E\sup_{t\in[0,T]}(|\xi(t)|+|\xi_K(t)|)^2\right)\E\indicator{\tau_K<T}\right)^{1/2}\\
&\leq \left(2\left(\E\sup_{t\in[0,T]}|\xi(t)|^2+\sup_{K>0}\E\sup_{t\in[0,T]}|\xi_K(t)|^2\right)\E\indicator{\tau_K<T}\right)^{1/2}\\
&\rightarrow 0\qquad\hbox{as }K\rightarrow \infty\ .
\end{align*}
This shows that $\xi_K(t)$ converges to a solution of \eqref{eq:partialInfo_BSDE}. Moreover, this $\xi$ is unique, because if there is another solution $(\tilde \xi,\tilde\alpha)\in\Sspace_T^2(\Pspace_1)\times\Hspace_T^2(\Pspace_\dd)$ solving \eqref{eq:partialInfo_BSDE}, then $\E\sup_{t\in[0,T]}|\tilde \xi(t)-\xi(t)|\leq \E\sup_{t\in[0,T]}|\tilde \xi(t)-\xi_K(t)|+\E\sup_{t\in[0,T]}|\xi(t)-\xi_K(t)|\rightarrow 0$ as $K\rightarrow\infty$, which shows that $\tilde \xi=\xi$ almost surely. 

Finally, uniqueness of $\alpha\in \Hspace_T^2(\Pspace_\dd)$ is shown by contradiction. Recall the formula $\alpha(t)=\xi(t)\left(\theta(t)-\frac{1-\gamma}{\gamma}\left(\sigma^{-1}(\hat h(t)-\mathbf{r})\right)\right)$ from \eqref{eq:alpha}, and suppose \eqref{eq:partialInfo_BSDE} has another solution with $\tilde\alpha\in \Hspace_T^2(\Pspace_\dd)$ such that $\tilde\alpha\neq \alpha$. Uniqueness of $\xi$ was already shown, so it must be that 
\[\tilde \xi(t) = M(0)\exp\left(-\int_0^t\left(\frac{\beta(u,\tilde\alpha(u),\tilde\xi(u))}{\tilde\xi(u)}+\frac12\left\|\frac{\tilde\alpha(u)}{\tilde\xi(u)}\right\|^2\right)du+\int_0^t\frac{\tilde\alpha(u)^\tr }{\tilde\xi(u)}d\zeta(u)\right)=\xi(t)\ ,\]
almost surely for all $t\in[0,T]$. Moreover, there is the process 
\[\tilde M(t) \triangleq Z(t)^{-\frac{1-\gamma}{\gamma}} \tilde \xi(t)=Z(t)^{-\frac{1-\gamma}{\gamma}} \xi(t)=M(t)\ .\]
Then from It\^o's lemma,
\[d\tilde M(t) = \tilde M(t)\left(\frac{\tilde\alpha(t)}{\tilde\xi(t)}+\frac{1-\gamma}{\gamma}\left(\sigma^{-1}(\hat h(t)-\mathbf{r})\right)\right)^\tr d\zeta(t)=M(t)\theta(t)^\tr d\zeta(t)=dM(t)\ ,\]
but $\theta$ is the unique martingale representation for $M(t)$ in the space $\Hspace_T^2(\Pspace_\dd)$ (see Proposition \ref{prop:thetaIntegrability} for proof that any $\theta$ is in $\Hspace_T^2(\Pspace_\dd)$), and so 
\begin{align*}
\tilde\alpha(t) &=\tilde\xi(t)\left( \theta(t)-\frac{1-\gamma}{\gamma}\left(\sigma^{-1}(\hat h(t)-\mathbf{r})\right)\right)\\
&=\xi(t)\left( \theta(t)-\frac{1-\gamma}{\gamma}\left(\sigma^{-1}(\hat h(t)-\mathbf{r})\right)\right)\\
&=\alpha(t)\ ,
\end{align*}
almost-surely for all $t\in[0,T]$.
\end{proof}

\subsection{The Partial-Information Optimal Strategy}
\label{sec:optimalStrategy}

Let $\pi^*$ denote the optimal strategy. From equation \eqref{eq:Vansatz}
\[V(t,x)=\E\left[U\left(X^{\pi^*}(T)\right)\Big|\Fpartial_t\vee\left\{X^{\pi^*}(t)=x\right\}\right]= U\left(xe^{r(T-t)}\right)G(t)\ ,\] 
and the process $V(t,X^{\pi^*}(t))$ is a true martingale. For any strategy $\pi\in\mathcal A$ the process $V(t,\Xpi(t))$ is a supermartingale, for which an SDE can be computed and the optimal strategy chosen so that the SDE has zero drift. This approach to finding the optimal $\pi^*$ yields the same optimum as found in \cite{elKarouiRouge2000} and \cite{imkeller2005}, and is the method used to prove the following result,
\begin{theorem}
Let $\Sigma=\sigma\sigma^\tr$. The optimal strategy is
\begin{equation}
\label{eq:optimalPi}
\pi^*(t) =\Sigma^{-1}\frac{\hat h(t)-\mathbf{r}}{\gamma}+(\sigma^{-1})^\tr\frac{\alpha(t)}{\xi(t)}\ ,
\end{equation}
where $\Sigma^{-1}\frac{\hat h(t)-\mathbf{r}}{\gamma}$ is the so-called \textbf{myopic} strategy and $(\sigma^{-1})^\tr\alpha(t)/\xi(t)$ is a dynamic hedging component due to stochasticity in the drift (see \cite{detempleGarcia2003,merton1971}).
\end{theorem}
\begin{proof}
Due to the properties of power utility, notice that $(1-\gamma)V(t,x)\geq 0$ for all $x\geq 0$ and all $\gamma>0$, $\gamma\neq 1$.

For any $\pi\in\mathcal A$ the SDE for $V(t,\Xpi(t))$ is
\begin{align}
\nonumber
&dV(t,\Xpi(t))\\
\nonumber
&=d\left(U(\Xpi(t)) e^{r(1-\gamma)(T-t)}\xi(t)^{\gamma}\right)\\
\nonumber
&=V(t,\Xpi(t))\Bigg((1-\gamma)\pi(t)^\tr(\hat h(t)-\mathbf{r})-\frac{\gamma(1-\gamma)\|\sigma^\tr\pi(t)\|^2}{2}+\gamma(1-\gamma)\pi(t)^\tr\sigma\frac{\alpha(t)}{\xi(t)}\\
\nonumber
&\hspace{7cm}-\left(\gamma\frac{ \beta(t,\alpha(t),\xi(t))}{\xi(t)}-\frac{\gamma(\gamma-1)}{2}\left\|\frac{\alpha(t)}{\xi(t)}\right\|^2\right)\Bigg)dt\\
\nonumber
&\hspace{6cm}+V(t,\Xpi(t))\left((1-\gamma)\pi(t)^\tr\sigma+\gamma \frac{\alpha(t)^\tr }{\xi(t)}\right)d\zeta(t)\\
\nonumber
&\leq(1-\gamma)\gamma V(t,\Xpi(t))\sup_{\pi(t)}\Bigg(\pi(t)^\tr\frac{\hat h(t)-\mathbf{r}}{\gamma}-\frac{\|\sigma^\tr\pi(t)\|^2}{2}+\pi(t)^\tr\sigma\frac{\alpha(t)}{\xi(t)}\\
\nonumber
&\hspace{7cm}-\left(\frac{1}{1-\gamma} \frac{\beta(t,\alpha(t),\xi(t))}{\xi(t)}+\frac{1}{2}\left\|\frac{\alpha(t)}{\xi(t)}\right\|^2\right)\Bigg)dt\\
\nonumber
&\hspace{6cm}+V(t,\Xpi(t))\left((1-\gamma)\pi(t)^\tr\sigma+\gamma \frac{\alpha(t)^\tr}{\xi(t)} \right)d\zeta(t)\\
\label{eq:dZeta_integral}
&\hspace{.5cm}=V(t,\Xpi(t))\left((1-\gamma)\pi(t)^\tr\sigma+\gamma \frac{\alpha(t)^\tr}{\xi(t)} \right)d\zeta(t)\ .
\end{align}
The maximized dt term is obtained by maximizing the quadratic form,
\begin{align}
\nonumber
&\pi^*(t)=\arg\max_{\pi(t)}\Bigg(-\|\sigma^\tr\pi(t)\|^2+2\left(\sigma^{-1}\frac{\hat h(t)-\mathbf{r}}{\gamma}+\frac{\alpha(t)}{\xi(t)}\right)^\tr\sigma^\tr\pi(t)\\
\label{eq:quadraticForm_partialProof}
&\hspace{8cm}-2\frac{\beta(t,\alpha(t),\xi(t))}{(1-\gamma)\xi(t)}-\left\|\frac{\alpha(t)}{\xi(t)}\right\|^2\Bigg)\ ,
\end{align}
from which first-order conditions yield $\pi^*(t)$ shown in \eqref{eq:optimalPi}. This maximizer is written in terms of $\alpha(t)$, the filter $\hat h(t)$, and the model parameters, and it is straightforward to check that the right-hand side of \eqref{eq:quadraticForm_partialProof} is equal to zero when evaluated at $\pi(t)=\pi^*(t)$ with $\beta(t,\alpha(t),\xi(t))$ given by Theorem \ref{thm:partialInfo_BSDE}. Hence, $V(t,X^{\pi^*}(t))$ is a supermartingale, and if it can be shown to be a true martingale then it is verified that $\pi^*$ is an optimal strategy (see \cite{zheng2011}).

Inserting the expression \eqref{eq:optimalPi} for $\pi^*(t)$ into \eqref{eq:dZeta_integral}, and then using expression \eqref{eq:alpha} for $\alpha(t)$ in terms of $\theta(t)$, there is the SDE 

\begin{align*}
dV(t,X^{\pi^*}(t))&=V(t,X^{\pi^*}(t))\left((1-\gamma)\pi^*(t)^\tr\sigma+\gamma \frac{\alpha(t)^\tr}{\xi(t)} \right)d\zeta(t)\\
&=V(t,X^{\pi^*}(t))\theta(t)^\tr d\zeta(t)\ ,
\end{align*}
where $\theta(t)$ is the martingale representation from \eqref{eq:M-martingaleRep}. Solving this SDE yields
\begin{align*}
V(t,X^{\pi^*}(t))& = V(0,X^{\pi^*}(0))\exp\left(-\frac12\int_0^t\|\theta(u)\|^2du+\int_0^t\theta(u)^\tr d\zeta(u)\right)\\
&=V(0,X^{\pi^*}(0))\frac{M(t)}{M(0)}\ ,
\end{align*}
which is a true martingale because $M(t)$ is a true martingale. Hence,
\begin{align*}
&\E\left[V(T,X^{\pi^*}(T))\Big|\Fpartial_t\vee \{X^{\pi^*}(t)=x\}\right]\\
&= V(t,X^{\pi^*}(t))+\underbrace{\E\left[\int_t^TV(u,X^{\pi^*}(u))\theta(u)^\tr d\zeta(u)\Big|\Fpartial_t\vee \{X^{\pi^*}(t)=x\}\right]}_{=0}\\
&=V(t,X^{\pi^*}(t))\ .
\end{align*}
This verifies that $\pi^*$ is an optimal strategy.
\end{proof}

\subsection{The Full-Information Value Function}
\label{sec:fullInformation}

Investment under `full information' means that the information in $\F_t$ is available to market participants; there are no hidden states because $(W(u),B(u))_{u\leq t}\in\F_t$. With full information the wealth process is 

\begin{align}
\nonumber
\frac{d\Xpi(t)}{\Xpi(t)}&=rdt+\sum_{i=1}^\dd\pi^i(t)(h^i(Y(t))-r)dt\\
\nonumber
&+ \sum_{i=1}^\dd\sum_{j=1}^{\mm} \pi^i(t)\sigma_\w^{ij}dW^j(t)+ \sum_{i=1}^\dd\sum_{j=1}^{\qq} \pi^i(t)\sigma_\y^{ij}dB^j(t)\ ,
\end{align} 
where $\pi$ is selected from among the set of full-information strategies
\begin{equation}
\label{eq:Afull}
 \mathcal A^{\full }=\left\{\F_t\hbox{-adapted}~~\pi:[0,T]\times\Omega\rightarrow \mathbb R^\dd~~\hbox{s.t.}~~ \int_0^T\Big|\Xpi(t)\|\pi(t)\|\Big|^2dt<\infty~~\hbox{a.s.}\right\}\ .
\end{equation}
Then the optimal investment is a Markov control problem,

\begin{align}
V^\full(t,x,y) 
\label{eq:Vfull}
&= \sup_{\pi\in\mathcal A^\full}\E\left[U\left(X(T)\right)\Big|X(t) = x,Y(t) = y\right]\ .
\end{align} 

\begin{proposition}
\label{prop:nirvanaFull}
Given \eqref{eq:novikov}, investor nirvana cannot occur in the full-information case for $\gamma>1$.
\end{proposition}
\begin{proof} 
The market is incomplete but the Novikov condition in \eqref{eq:novikov} means that a possible equivalent martingale measure is the one having Radon-Nikodym derivative
\begin{align}
\nonumber
&\mathcal E(t) = \exp\left( -\frac12\int_0^t\left\|\sigma^{-1}(h(Y(u))-\mathbf{r} )\right\|^2 du\right.\\
\nonumber
&\hspace{3cm}\left.-\int_0^t(h(Y(u))-\mathbf{r} )^\tr\Big((\sigma_\w^{-1})^\tr dW(u)+(\sigma_\y^{-1})^\tr dB(u)\Big)\right) \ ,
\end{align}
i.e., the minimal-entropy martingale measure. Now, it should be clear that $\mathcal E(t)$ can be non-zero, namely $\mathbb P\left(\mathcal E(T)/\mathcal E(t)>0\Big|Y(t)=y\right)>0$, and so

\[\E\left[\mathcal E(T)^{\frac{\gamma-1}{\gamma}}\Big|Y(t)=y\right]>0\ ,\]
from which it follows that the full-information value function has the following duality bound:
 \begin{align*}
 V^\full(t,x,y)&\leq \inf_p\left(\E\left[U^\star\left(pe^{-r(T-t)}\frac{\mathcal E(T)}{\mathcal E(t)}\right)\Big|Y(t)=y\right]+xp\right)\\
 & = \inf_p\left(U^\star(pe^{-r(T-t)})\E\left[\left(\frac{\mathcal E(T)}{\mathcal E(t)}\right)^{\frac{\gamma-1}{\gamma}}\Big|Y(t)=y\right]+xp\right)\\
 &<0\ .
 \end{align*}
Hence, nirvana in the sense of Definition \ref{def:nirvana} does not occur.
\end{proof}
The full-information value function satisfies a Hamilton-Jacobi-Bellman (HJB) equation,

\begin{align}
\nonumber
\left(\frac{\partial}{\partial t}+rx\frac{\partial}{\partial x}+\mathcal L\right)V^\full\hspace{9.5cm}\\
\label{eq:HJB}
+\sup_\pi\left(\frac{x^2}{2}\pi^\tr\Sigma\pi\frac{\partial^2}{\partial x^2}V^\full+x\pi^\tr( h(y)-\mathbf{r})\frac{\partial}{\partial x}V^\full+x\pi^\tr\sigma_\y a(y)^\tr\frac{\partial}{\partial x}\nabla V^\full\right)&=0\\
\nonumber
V^\full\Big|_{t=T}&=U\ ,
\end{align}
where $\Sigma=\sigma\sigma^\tr$, $\nabla $ denotes the gradient in $y$, and
\[\mathcal L=\frac12\sum_{i,j=1}^{\qq}\left(aa^\tr(y)\right)^{ij}\frac{\partial^2}{\partial y_i\partial y_j}+\sum_{i=1}^{\qq}b^i(y)\frac{\partial}{\partial y_i}\ .\]
If \eqref{eq:HJB} has a classical solution then the optimal strategy is written in feedback form,
\begin{equation}
\label{eq:piStarFull}
\pi^*(t,x,y)=-\Sigma^{-1}\Bigg(\left(h(y)-\mathbf{r}\right)\frac{\frac{\partial }{\partial x}V^\full(t,x,y)}{x\frac{\partial^2 }{\partial x^2}V^\full(t,x,y) }-\sigma_\y a(y)^\tr \frac{\frac{\partial}{\partial x}\nabla V^\full(t,x,y)}{x\frac{\partial^2 }{\partial x^2}V^\full(t,x,y) }\Bigg) \ .
\end{equation}
By Theorem 8.1 in Chapter III.8 of \cite{flemingSonerBook}, if $\pi^*$ given by \eqref{eq:piStarFull} is an admissible strategy in $\mathcal A^\full$, then strict concavity of the objective inside the supremum implies that a classical solution to \eqref{eq:HJB} will satisfy a verification lemma.

For the case of power utility there is a simplifying ansatz for the solution to equation \eqref{eq:HJB}, 

\begin{equation}
\label{eq:fullInfoAnsatz}
V^\full(t,x,y) = U\left(xe^{r(T-t)}\right) G^\full(t,y)\ ,
\end{equation}
which means $G^\full$ satisfies the equation
\begin{align}
\label{eq:HJB_G}
\left(\frac{\partial}{\partial t}+\mathcal L\right)G^\full+(1-\gamma)\max_{\pi\in\mathbb R^\dd} f\left(y,\pi,G^\full,a^\tr\nabla G^\full\right)=0&\\
\nonumber
G^\full\Big|_{t=T}=1&\ ,
\end{align}
where the objective function $f$ is strictly concave in $\pi$ for any $(y,\pi,g,\eta)\in\mathbb R^\qq\times\mathbb R^\dd\times\mathbb R^+\times\mathbb R^\qq$, as
\begin{align}
\label{eq:f_concave}
&f\left(y,\pi,g,\eta\right)=\left(-\frac\gamma2\pi^\tr\Sigma\pi+\pi^\tr(h(y)-\mathbf{r})\right)g+\pi^\tr\sigma_\y \eta\ .
\end{align}
The objective in \eqref{eq:f_concave} can be maximized with first-order conditions, where the maximizer is 
\begin{equation}
\label{eq:piStarFull_BSDE}
\pi^*(t,y) = \Sigma^{-1}\left(\frac{h(y)-\mathbf{r}}{\gamma} +\sigma_\y\frac{\eta}{\gamma g}\right)\ ,
\end{equation}
from which it is seen that the maximized objective is
\begin{align}
\nonumber
F(y,g,\eta) &= \max_{\pi\in\mathbb R^\dd}f(y,\pi,g,\eta)\\
\nonumber
&=\frac{g}{2\gamma}\left((h(y)-\mathbf{r}) +\sigma_\y\frac{ \eta}{g}\right)^\tr\Sigma^{-1}\left((h(y)-\mathbf{r}) +\sigma_\y \frac{\eta}{g}\right)\\
\label{eq:F}
&\geq 0\ .
\end{align}
If equation \eqref{eq:HJB_G} has a classical solution, then an optimal strategy is found by evaluating \eqref{eq:piStarFull_BSDE} at $(g,\eta)=(G^\full,a^\tr\nabla G^\full)$,
\[\pi^*(t,y)=\Sigma^{-1}\Bigg(\frac{h(y)-\mathbf{r}}{\gamma} +\frac{1}{\gamma G^\full(t,y)}\sigma_\y a(y)^\tr\nabla G^\full(t,y)\Bigg) \ ,\]
which can be seen as being comprised of two components: a myopic component given by the optimal from the standard Merton problem, plus a dynamic hedging term motivated by stochastic fluctuations in $Y(t)$.

\begin{remark}[Examples of Other Nonlinear HJB Equations]
Some examples in the finance literature where there occurs a nonlinear HJB equation like \eqref{eq:HJB_G} include: optimal portfolio allocation with consumption and an unhedgeable income stream \cite{duffieFlemingSonerZaripho1997}; a generalization of problem \eqref{eq:Vfull} but with scalar $Y(t)$ in \cite{stoikovZaripho2005}. Other examples include the linear case (i.e., $h(y)$ and $b(y)$ linear, $a(y)$ constant in $y$) where the solution to \eqref{eq:HJB_G} can be found with an affine ansatz (see \cite{bensoussanBook,brendle2006}); these linear models can have investor nirvana if there is low risk aversion (see \cite{kimOmberg} or Section \ref{sec:nirvana} of this paper). 
\end{remark}

Equation \eqref{eq:HJB_G} is a semi-linear PDE with uniformly elliptic operator, for which classical solutions have been shown to exist under relatively general circumstances. Existence of smooth solutions are shown \cite{pham2002}, and for scalar cases it is shown in \cite{zariphopoulou2001} that the PDE for $G^\full$ reduces to a power transform of a solution to a linear PDE. Specifically, for the case of $a(y)$ constant in $y$, \cite{pham2002} gives a sufficient condition for smooth solutions to the HJB,

\begin{condition}
\label{cond:cons_a_exist}
If the diffusion matrix $a$ in equation \eqref{eq:dY} is constant in $y$, with
\begin{align*}
&b(y)~\hbox{and }h(y)~\hbox{being $C^1$ and Lipschitz in $y$, and}\\
&\|\sigma^{-1}h(y)\|^2~\hbox{being $C^1$ and Lipschitz in $y$,}
\end{align*} 
then there exists a function $\varphi(t,y)$ differentiable in $t$ and twice differentiable in $y$ such that 
\[G^{\full}(t,y) = \exp(-\varphi(t,y))\ ,\]
i.e., there is a classical solution to equation \eqref{eq:HJB_G}, and with $|\nabla \varphi(t,y)|\leq C(1+|y|)$ for all $t\in[0,T]$ and for all $y\in\mathbb R^\qq$.
\end{condition}

\begin{remark}
For non-constant  $a(y)$, \cite{pham2002} explains how to reparameterize the SDE for $Y(t)$ so that the Condition \ref{cond:cons_a_exist} applies, namely by looking for a function $\phi(y)$ with 

\[\nabla \phi(y) = a(y)^{-1}\qquad\hbox{i.e., the inverse of matrix $a(y)$}\ ,\]
so that $ Y(t) = \phi^{-1}(\tilde  Y(t))$ with
\[d\tilde Y(t) = \tilde b(\tilde Y(t))dt +dB(t)\ ,\]
where $\tilde b(\tilde y) = \Big(\nabla\phi( y)^\top  b( y) + \frac12\hbox{trace}\left[a(y)^\top\left(\nabla\nabla^\top\phi( y)\right)a(y)\right]\Big)\Big|_{y=\phi^{-1}(\tilde y)}$. From here it must be checked that $\tilde b$ is $C^1$ and Lipschitz.
\end{remark}

The solution $G^\full$ is the value function
\begin{align*}
&G^\full(t,y)=1+(1-\gamma)\\
&{\small\times\sup_{\pi\in\mathcal A^\full}\E\left[\int_t^Tf\Big(Y(u),\pi(u),G^\full(Y(u)),a(Y(u))^\tr\nabla G^\full(u,Y(u))\Big)du\Bigg|Y(t)=y\right]}\ .
\end{align*}
As explained on page 143 in Chapter 6 of \cite{pham}, there is a nonlinear Feynman-Kac representation for $G^\full$, with $G^\full(t,Y(t))=\goodchi(t)$ where $\goodchi(t)$ solves the following BSDE,
\begin{align}
\nonumber
-d\goodchi(t) &=(1-\gamma) F\left(Y(t),\goodchi(t) , \psi(t)\right)dt-\psi(t)^\tr dB(t)\ ,\qquad\hbox{for } t\leq T\\
\label{eq:fullInfo_BSDE}
\goodchi(T)&=1\ .
\end{align}
A solution to \eqref{eq:fullInfo_BSDE} is a pair $(\goodchi,\psi)\in\Sspace_T^2(\Pspace_1^\full)\times \Hspace_T^2(\Pspace_{\qq}^\full)$ with 
\begin{align*}
\Pspace_\qq^\full&=\Big\{\hbox{the set of $\qq$-dimensional $\mathcal F_t^B$-adapted measurable processes on $\Omega\times[0,T]$}\Big\}\ ,
\end{align*}
where $\Sspace_T^2$ and $\Hspace_T^2$ are the same as those defined in \eqref{eq:spaces} except with $\Pspace_\qq^\full$. Given the solution to \eqref{eq:fullInfo_BSDE}, the optimal strategy is
\begin{equation}
\label{eq:optimalPifull}
\pi^*(t) =\pi^*(t,Y(t),\goodchi(t),\psi(t))=  \Sigma^{-1}\left(\frac{h(Y(t))-\mathbf{r}}{\gamma}+\sigma_\y\frac{\psi(t)}{\gamma\goodchi(t)}\right)\ ,
\end{equation}
which is similar to the formula in \eqref{eq:piStarFull_BSDE}, and is an admissible strategy (i.e., is S-integrable) because $\goodchi(t)>0$ a.s. by a comparison principle as explained in Theorem 6.2.2 on page 142 in Chapter 6 of \cite{pham}.

Existence of solutions to \eqref{eq:fullInfo_BSDE} are not covered by the general theory in \cite{reneSIAM,elKarouiPengQuenez1997,pham} because $F(t,y,g,p)$ does not have a uniform Lipschitz constant, and is not covered by \cite{kobylanski2000} because $F^2$ has a $g^2 $ term. However, a classical solution $G^\full$ to \eqref{eq:HJB_G} can be evaluated at $Y(t)$ to obtain the solution to the BSDE,
\begin{align}
\label{eq:fullInfo_BSDEsolution}
\goodchi(t)&=G^\full(t,Y(t))\\
\nonumber
\psi(t)&=a(Y(t))^\tr\nabla G^\full(t,Y(t))\ ,
\end{align}
provided that this solution is in $\Sspace_T^2(\Pspace_1^\full)\times \Hspace_T^2(\Pspace_{\qq}^\full)$.

\begin{proposition}
\label{prop:fullInfoBSDEcond}
Suppose Condition \ref{cond:cons_a_exist}.  If
\begin{equation}
\label{eq:polynomialMoments}
\E\exp\left(\frac{2\delta_1|\gamma-1||\gamma-2|}{\epsilon\gamma^2}\int_0^T\|h(Y(t))\|^2dt\right)<\infty\ ,
~~~~\hbox{and}~~~~\E\int_0^T\|Y(t)\|^{2\delta_2}dt<\infty\ ,T
\end{equation}
for some $\delta_1,\delta_2>1$ with $\frac{1}{\delta_1}+\frac{1}{\delta_2}=1$, then the pair given by equation \eqref{eq:fullInfo_BSDEsolution} is in $\Sspace_T^2(\Pspace_1^\full)\times \Hspace_T^2(\Pspace_{\qq}^\full)$, and hence a solution to BSDE \eqref{eq:fullInfo_BSDE}.
\end{proposition}
\begin{proof}
If Condition \ref{cond:cons_a_exist} holds then the gradient of $\log G^{\full}$ has a linear growth bound, and hence the integrability condition
\begin{align}
\nonumber
&\E\int_0^T\|a(Y(t))^\tr\nabla  G^\full(t,Y(t))\|^2dt\\
\nonumber
&\leq C^2\E\int_0^T|  G^\full(t,Y(t))|^2(1+\|Y(t)\|)^2dt \\
\label{eq:polynomialBound}
&\leq C^2\left(\E\int_0^T|  G^\full(t,Y(t)))|^{2\delta_1}dt\right)^{1/\delta_1}\left(\E\int_0^T\|1+Y(t)\|^{2\delta_2}dt\right)^{1/\delta_2}\ ,
\end{align}
where $\delta_1,\delta_2 \geq 1$ with $\tfrac{1}{\delta_1}+\tfrac{1}{\delta_2}=1$. From the duality bound 
\[V^\full(t,x,y)\leq U^\star(pe^{-r(T-t)})\E\left[\left(\frac{\mathcal E(T)}{\mathcal E(t)}\right)^{\frac{\gamma-1}{\gamma}}\Big|Y(t)=y\right]+xp \ ,\]
we have $\E\sup_{t\in[0,T]} |G^\full(t,Y(t)))|^{2\delta_1}<\infty$ if $\E\sup_{t\in[0,T]}\left(\E\left[\left(\frac{\mathcal E(T)}{\mathcal E(t)}\right)^{\frac{\gamma-1}{\gamma}}\Big|Y(t)\right]\right)^{2\delta_1}<\infty$, and so taking steps similar to those in the proof of Proposition \ref{prop:sup_xiK_bound} it follows that a sufficient condition for finiteness of inequality \eqref{eq:polynomialBound} are the inequalities of \eqref{eq:polynomialMoments}; because $\delta_1\geq 1$ it follows from the first inequality of \eqref{eq:polynomialMoments} that $G^\full(t,Y(t)))\in\Sspace_T^2(\Pspace_1^\full)$.
\end{proof}
In the literature, Proposition 6.3.2 in Chapter 6.3 of \cite{pham} shows equation \eqref{eq:fullInfo_BSDEsolution} to be in $\Sspace_T^2(\Pspace_1^\full)\times \Hspace_T^2(\Pspace_{\qq}^\full)$ if $G^\full(t,y)$ has at most linear growth in $y$ and  if the gradient has a bound of polynomial growth $\|a(y)^\tr\nabla G^\full(t,y)\| \leq C(1+\|y\|^n)$ for some $C\geq 0$ and $n\geq 0$.

\begin{proposition}
\label{prop:verification}
If a unique solution to \eqref{eq:fullInfo_BSDE} exists, then $\pi^*(t)=\pi^*(t,Y(t),\goodchi(t),\psi(t))$ given by \eqref{eq:optimalPifull} is such that $U(X^{\pi^*}(t))\goodchi(t)$ satisfies a verification lemma, and hence $\pi^*$ is the optimal strategy.
\end{proposition}
\begin{proof} (See Appendix \ref{app:fullInfoVerification}).\end{proof}

If there exists a solution to BSDE \eqref{eq:fullInfo_BSDE} then it is unique:

\begin{theorem}
\label{thm:fullInfo_uniqueness}
If there exists $(\goodchi,\psi)\in\Sspace_T^2(\Pspace_1^\full)\times \Hspace_T^2(\Pspace_{\qq}^\full)$ that is a solution to BSDE \eqref{eq:fullInfo_BSDE}, then it is also the unique solution.
\end{theorem}
\begin{proof}(See Appendix \ref{app:fullInfo_BSDE}).\end{proof}

\begin{remark}
Condition \ref{cond:cons_a_exist}, Proposition \ref{prop:fullInfoBSDEcond}, Proposition \ref{prop:verification}, and Theorem \ref{thm:fullInfo_uniqueness} contributed toward existence of BSDE solutions to solve the full-information control problems. Sections \ref{sec:nirvana} and \ref{sec:nonlinearExample} provide financial examples with explicit formulae for classical solutions.
\end{remark}

\begin{remark}[Existence in the Absence of Classical Solutions]
\label{remark:viscositySolutions}
The solution to \eqref{eq:fullInfo_BSDE} can exist without the existence of a classical solution to \eqref{eq:HJB_G}. A solution $(\goodchi,\psi)$ has associated with it a viscosity solution to \eqref{eq:HJB_G}, i.e., there is a deterministic function $G^\full$ such that
\[\goodchi(t)=G^\full(t,Y(t))\qquad\hbox{almost surely, where $G^\full$ is a viscosity solution of \eqref{eq:HJB_G},}\]
(see Proposition 6.3.3 in Chapter 6.3 of \cite{pham}). However, existence of a viscosity solution may not be sufficient for existence of a solution to \eqref{eq:fullInfo_BSDE}, as (i) $G^\full$ must be square integrable and (ii) it is not clear how to construct $\psi$ from the viscosity solution. Moreover, a solution to \eqref{eq:fullInfo_BSDE} might be identified if the viscosity solution is unique, but the current theory for uniqueness of viscosity solutions requires the PDE to satisfy a strong comparison principle and also some growth conditions \cite{kobylanski2000,pham} that are not satisfied by the nonlinear term $F(t,y,g,\eta)$. Lastly, it should be pointed out that the operator $\mathcal L$ is what is called \textit{degenerate elliptic}, and so a classical solution to \eqref{eq:HJB_G} is also a viscosity solution (see \cite{CL1992}), reaffirming that \eqref{eq:fullInfo_BSDEsolution} is the appropriate formula if there is regularity.

\end{remark}

\subsection{The Information Premium}
\label{sec:premium}
Intuitively it would seem that full information is better than partial --or at least that it cannot hurt investment. This is correct, but the full-information market is incomplete because $Y(t)$ is not tradeable, and cannot be reduced to a complete market like that given in \eqref{eq:dScomplete}. Generally speaking, there is added premium and lowered utility when a model is incomplete. However, partial information is an exception, as it turns out that the partially-informed investor expects the fully informed to have an advantage.

From the perspective of the partially-informed investor, the information premium (i.e., the loss in utility due to partial information) is,
\[\premium(t,x)\triangleq\E\left[V^\full(t,x,Y(t))-V(t,x)\Big|\Fpartial_t\right]=U(xe^{-r(T-t)})\E\left[G^\full(t,Y(t))-G(t) \Big|\Fpartial_t\right]\ .\]
This is similar to the loss of information quantified in \cite{brendle2006,carmonaIP4} for the linear Gaussian problem, but is quantified with BSDEs for the general nonlinear case.
\begin{proposition}
\label{prop:infoPremium}
The information premium is equal to
\begin{align}
\nonumber
&\premium(t,x)\\
\nonumber
&=(1-\gamma)U(xe^{r(T-t)})\\
\nonumber
&\times\E\Bigg[\int_t^T\Bigg(F\left(Y(u),\goodchi(u),\psi(u)\right)-\underbrace{\gamma\Bigg(\frac{\beta(u,\alpha(u),\xi(u))}{(1-\gamma)\xi(t)}+\frac{1}{2}\left\|\frac{\alpha(u)}{\xi(u)}\right\|^2\Bigg)G(u)}_{\geq0}\Bigg)du \Bigg|\Fpartial_t\Bigg]\\
\label{eq:premium}
&\geq0\ ,
\end{align}
where $(1-\gamma)U(x)\geq 0$ by definition for all $x\geq 0$.
\end{proposition}
\begin{proof}
The fully-informed investor has the option to follow the partially-informed optimal strategy, hence,
\begin{align}
\nonumber
&\E\left[V^\full\left(t,x,Y(t)\right)\Big|\Fpartial_t\vee\{X^\pi(t)=x\}\right]\\
\nonumber
&=\E\left[\sup_{\pi\in\mathcal A^\full}\E\left[U(X^\pi(T))\Big|\F_t\vee\{X^\pi(t)=x\}\right]\Big|\Fpartial_t\vee\{X^\pi(t)=x\}\right]\\
\nonumber
&\geq\E\left[\sup_{\pi\in\mathcal A}\E\left[U(X^\pi(T))\Big|\F_t\vee\{X^\pi(t)=x\}\right]\Big|\Fpartial_t\vee\{X^\pi(t)=x\}\right]\\
\nonumber
&\geq\sup_{\pi\in\mathcal A}\E\left[\E\left[U(X^\pi(T))\Big|\F_t\vee\{X^\pi(t)=x\}\right]\Big|\Fpartial_t\vee\{X^\pi(t)=x\}\right]\\
\label{eq:EVgeqV}
&=V(t,x)\ ,
\end{align}
for all $t\in[0,T]$ and all $x\geq0$, 
and
\[\premium(t,x)\geq0\qquad\hbox{for all $x>0$ and $t\in[0,T)$.}\]
Using the BSDEs of \eqref{eq:fullInfo_BSDE} and the inequality shown in \eqref{eq:EVgeqV}, the information premium is written as
\begin{align*}
&\premium(t,x)\\
&=U(xe^{r(T-t)})\E\left[G^\full(t,Y(t))-G(t) |\Fpartial_t\right]\\
&=(1-\gamma)U(xe^{r(T-t)})\\
&\times\E\Bigg[\int_t^T\Bigg(F\left(Y(u),\goodchi(u),\psi(u)\right)-\gamma\underbrace{\Bigg(\frac{\beta(u,\alpha(u),\xi(u))}{(1-\gamma)\xi(t)}+\frac{1}{2}\left\|\frac{\alpha(u)}{\xi(u)}\right\|^2\Bigg)}_{\geq0}G(u)\Bigg)du \Bigg|\Fpartial_t\Bigg]\\
&\geq0\ ,
\end{align*}
where $(1-\gamma)U(x)\geq 0$ by definition for all $x\geq 0$, $F\left(Y(t),\goodchi(t),\psi(t)\right)\geq 0$ for all $t\in[0,T]$, and
\[\frac{\beta(t,\alpha(t),\xi(t))}{(1-\gamma)\xi(t)}+\frac{1}{2}\left\|\frac{\alpha(t)}{\xi(t)}\right\|^2~=~~\frac{1}{2}\left\|\sigma^{-1}\frac{\hat h(t)-\mathbf{r}}{\gamma}+\frac{\alpha(t)}{\xi(t)}\right\|^2\geq 0\]
by  the formula for $\beta(t,\alpha(t),\xi(t))$ given in Theorem \ref{thm:partialInfo_BSDE}. 
\end{proof}
The importance of \eqref{eq:premium} is that it shows how the information premium incrementally grows with time. Alternatively, one could look at $d\left(\E\left[G^\full(t,Y(t))-G(t) |\Fpartial_t\right]\right)$, but the BSDEs provide a different perspective because the coefficients provide a breakdown of the premium's growth.

Before moving to the next section, it should be pointed out how the information premium can be either infinite or undefined. The obvious lower bound $V(t,x)\geq U(xe^{r(T-t)})$ is obtained with $\pi\equiv 0$, and leads to the upper bound
\begin{align}
\nonumber
\premium(t,x)&\leq U\left(xe^{r(T-t)}\right)\left(\E\left[G^\full\left(t,Y(t)\right) \Big|\Fpartial_t\right]-1\right)\ .
\end{align}
These bounds depend on finiteness of the full-information value function, and so investor nirvana for full information occurring with non-zero probability results in either
\begin{itemize}
\item $\premium(t,x)=\infty$ because $V(t,x)<\infty$  and $\E[ V^\full(t,x,Y(t))|\Fpartial_t]=\infty$,
\item $\premium(t,x)=\infty-\infty$ (undefined) because  $V(t,x)=\infty$  and $\E[ V^\full(t,x,Y(t))|\Fpartial_t]=\infty$.
\end{itemize}
These two cases are considered at the end of Section \ref{sec:nirvana}. It should also be pointed out that the information premium is usually positive, as shown numerically in \cite{FPS2015,FPS2016,pap2013}.

\section{The Linear Case}
\label{sec:nirvana}
Consider the linear case with $h(y)=\mu+y$. Suppose that $Y(t)\in\mathbb R^1$ is an Ornstein-Uhlenbeck process, and there is only one risky asset so that $S(t)\in\mathbb R^1$. The SDEs are
\begin{align}
\label{eq:linearModel_dS}
\frac{dS(t)}{S(t)} &= (\mu+Y(t))dt+\sigma \left(\sqrt{1-\rho^2}dW(t)+\rho dB(t)\right)\\
\label{eq:linearModel_dY}
dY(t)&=-\kappa Y(t)dt+ a dB(t)\ ,
\end{align}
with $\kappa, a,\sigma>0$, $\rho\in(-1,1)$, and $\mu\in\mathbb R$ being the long-term mean rate of return. The wealth process is
\begin{align*}
\frac{d\Xpi(t)}{\Xpi(t)}&= rdt+\pi(t)\left(\frac{dS(t)}{S(t)}-rdt\right)\\
&=\Big(\pi(t)(\mu+Y(t))+r(1-\pi(t))\Big)dt+\pi(t)\sigma \left(\sqrt{1-\rho^2}dW(t)+\rho dB(t)\right)\ . 
\end{align*}
For simplicity take $r=\mu=0$. This model is the scalar version of the model considered in \cite{brendle2006, carmonaIP4,wangWu2008}, except that they avoided nirvana situations by considering the case of $\gamma>1$. Indeed, this section considers $\gamma<1$ and examines the stability of a scalar Riccati equation, whereas stability of the matrix Riccati equation in \cite{brendle2006, carmonaIP4} would require a significantly more difficult analysis.

\subsection{The Fully-Informed Investor}
The optimal investment problem for full information is
\[ V(t,x,y)=\sup_\pi\mathbb E\left[U(X(T))\Big|X(t)=x,Y(t)=y\right]\ ,\]
which is the solution $V(t,x,y)$ to the HJB equation 
\begin{align}
\nonumber
V_t +\frac{ a^2}{2}V_{yy}-\kappa yV_y-\frac{\left(yV_x+\rho\sigma a V_{xy}\right)^2}{2\sigma^2 V_{xx}}&=0\\
\nonumber
V\Big|_{t=T}&=U\ ,
\end{align}
where the optimal portfolio is 
\[\pi^* = -\frac1x\frac{yV_x+\rho\sigma a V_{xy}}{\sigma^2 V_{xx}}\ .\] 
For power utility $U(x) = \frac{1}{1-\gamma}x^{1-\gamma}$ the solution of the HJB equation is given by the ansatz $V(t,x,y) = U(x)G(t,y)$, which yields the following equation for $G$:
\begin{align*}
G_t +\frac{ a^2}{2}G_{yy}-\kappa yG_y+\frac{1-\gamma}{\gamma}\frac{\left(yG+\rho\sigma a G_{y}\right)^2}{2\sigma^2 G}&=0\\
G\Big|_{t=T}&=1\ ,
\end{align*}
where 
\[\pi^* = \frac{y}{\gamma\sigma^2}+\frac{\rho a G_y}{\gamma\sigma G}\ .\]
We now apply another ansatz, 
\[G(t,y) = \exp\Big(A(t)y^2+H(t)\Big)\ ,\]
for which there are the ordinary differential equations
\begin{align}
\label{eq:aPrime}
A'(t)+2 a^2\left(1+\frac{(1-\gamma)\rho^2}{\gamma}\right)A^2(t)-2\left(\kappa-\frac{(1-\gamma)\rho a }{\gamma\sigma}\right)A(t)+\frac{1-\gamma}{2\gamma\sigma^2}&=0\\
\label{eq:H_Primes}
H'(t)+ a^2 A(t)&=0\ ,
\end{align}
with terminal conditions $A(T)=0=H(T)$ apply. Then the optimal control is
\[\pi^*(t) = \frac{y}{\gamma\sigma^2}+\frac{2\rho ay A(t)}{\gamma\sigma }\ .\]
Let $A_\pm$ be the roots of the polynomial $2 a^2\left(1+\frac{(1-\gamma)\rho^2}{\gamma}\right)A^2(t)-2\left(\kappa-\frac{(1-\gamma)\rho a }{\gamma\sigma}\right)A(t)+\frac{1-\gamma}{2\gamma\sigma^2}$. From the quadratic equation, these roots are found to be
\begin{equation}
\label{eq:aPlusMinus}
A_\pm = \frac{2\left(\kappa-\frac{(1-\gamma)\rho a }{\gamma\sigma}\right)\pm\sqrt{4\left(\kappa-\frac{(1-\gamma)\rho a }{\gamma\sigma}\right)^2-4\frac{(1-\gamma) a^2}{\gamma\sigma^2}\left(1+\frac{(1-\gamma)\rho^2}{\gamma}\right)}}{4 a^2\left(1+\frac{(1-\gamma)\rho^2}{\gamma}\right)}\ ,
\end{equation}
and the Riccati equation \eqref{eq:aPrime} is written as
 \begin{equation}
 \label{eq:rootRiccati}
A'(t) = -\frac{c}{2}(A(t)-A_+)(A(t)-A_-)\ ,
\end{equation}
where $c = 4 a^2\left(1+\frac{(1-\gamma)\rho^2}{\gamma}\right)$.

\subsubsection{Complex Roots and Nirvana Strategies}
\label{sec:fullInfoNirvana}
The roots $A_\pm$ given by equation \eqref{eq:aPlusMinus} are real iff 
\begin{align}
\nonumber
0&\leq \left(\kappa-\frac{(1-\gamma)\rho a }{\gamma\sigma}\right)^2-\frac{(1-\gamma) a^2}{\gamma\sigma^2}\left(1+\frac{(1-\gamma)\rho^2}{\gamma}\right)\\
\label{eq:expStability}
&= \kappa^2-\frac{(1-\gamma) a }{\gamma\sigma}\left(2\kappa\rho+\frac{ a}{\sigma}\right)\ .
\end{align}
Instabilities can arise if the roots are complex. The best way to understand why is to look at the linearization of the Riccati equation \eqref{eq:aPrime}. Letting $v(t)$ be the solution to the following linear equation,
\[v''-2\left(\kappa-\frac{(1-\gamma)\rho a }{\gamma\sigma}\right)v'+\frac{(1-\gamma) a^2}{\gamma\sigma^2}\left(1+\frac{(1-\gamma)\rho^2}{\gamma}\right)v=0\ ,\]
with the appropriate terminal conditions $v'(T)=0$ and $v(T)\neq 0$, the solution to Riccati equation \eqref{eq:aPrime} is $A(t) = v'(t)/(2 a^2 v(t))$. For a characteristic equation with complex roots, the solution is
\[v(t) = e^{\left(\kappa-\frac{(1-\gamma)\rho a }{\gamma\sigma}\right)(T-t)}\Big(C_1\cos(\Xi(T-t))+C_2\sin(\Xi(T-t))\Big)\ ,\]
where $\Xi$ is the absolute value of the imaginary component
\[\Xi =  \left|\sqrt{  \kappa^2-\frac{(1-\gamma) a }{\gamma\sigma}\left(2\kappa\rho+\frac{ a}{\sigma}\right)}\right|\ ,
\]
and where the constants are chosen to match the terminal conditions, so that $C_1\left(\kappa-\frac{(1-\gamma)\rho a }{\gamma\sigma}\right)=C_2$. Investor nirvana occurs because it may be that $v(t) =0 $ for some $t\in[0,T]$. If this is the case, then $A(t)$ will blow up at some finite time $0\leq t<T$. An example of such an instability is $2\kappa\rho+\frac{ a}{\sigma}>0$ and $\gamma$ tending toward zero. Another instability occurs for $\kappa$ tending toward zero with constant $\gamma\in(0,1)$.

\subsection{The Partially-Informed Investor}
Letting $\widehat Y(t) =\mathbb E[Y(t)|\Fpartial_t]$, the innovations process is
\[\nu(t) = \int_0^t\left(\frac{dS(u)}{S(u)}-\widehat Y(u)du\right)\ ,\]
for which $\frac1\sigma\nu(t)$ is a Brownian motion. Furthermore, letting $\Sigma(t)= \mathbb E(Y(t)-\widehat Y(t))^2$, the investor tracks the hidden process $Y(t)$ using the Kalman filter 
\begin{align*}
d\widehat Y(t)& = -\kappa \widehat Y(t)dt+\frac{1}{\sigma^2}\left( \Sigma(t)+\sigma a\rho\right)d\nu(t)\\
\frac{d}{dt}\Sigma(t)&=-2\kappa\left(\Sigma(t)-\frac{ a^2(1-\rho^2)}{2\kappa}\right)-\frac{2 a\rho}{\sigma}\Sigma(t)-\left(\frac1\sigma\Sigma(t)\right)^2\ ,
\end{align*}
where for $t$ large there is the asymptotic  $\Sigma(t)\rightarrow \overline\Sigma$ as $t\nearrow\infty $ with
\begin{align}
\label{eq:barSigma}
\overline\Sigma &= -(\kappa\sigma^2+ a\rho\sigma)+\sqrt{(\kappa\sigma^2+ a\rho\sigma)^2+\left( a\sigma\sqrt{1-\rho^2}\right)^2}\ ,
\end{align}
Assuming $\Sigma(0) = \overline\Sigma$, then $\frac{d}{dt}\Sigma(t)=0$ for all $t>0$ and the partial-information model is written with constant coefficients and the innovations,

\begin{align*}
\frac{dS(t)}{S(t)} &=\widehat Y(t)dt+ \sigma d\zeta(t)\\
d\widehat Y(t)&=-\kappa \widehat Y(t)dt+\bar a d\zeta(t)\ ,
\end{align*} 
where $\zeta(t)=\frac1\sigma\nu(t)$ is a Brownian motion and $\bar  a=\frac{1}{\sigma}\left( \overline\Sigma+\sigma a\rho\right)$. Hence, the partial-information model is equivalent to the full-information model in equations \eqref{eq:linearModel_dS} and \eqref{eq:linearModel_dY} having $\rho=1$ and diffusion coefficient $\bar a$. 

An investor achieves nirvana when $V(t,x,y)=\infty$ for $\gamma\in(0,1)$, and when $V(t,x,y)=0$ for $\gamma>1$ (see Definition \ref{def:nirvana} or \cite{kimOmberg}). Propositions \ref{prop:nirvana} and \ref{prop:nirvanaFull} showed nirvana cannot occur for $\gamma>1$ for both partial and full information, respectively. For partial-information this can be verified for the linear model by investigating the parameters. Similar to the condition set forth in \eqref{eq:expStability}, the partial-information ansatz involves a real root iff 
\begin{equation}
\label{eq:partialInfoStability} 
\kappa^2-\frac{(1-\gamma)\bar a }{\gamma\sigma}\left(2\kappa+\frac{\bar a}{\sigma}\right) \geq 0\ .
\end{equation}
For $\gamma>1$ a minimum of zero is achieved in \eqref{eq:partialInfoStability} when $\bar a = -\kappa\sigma$. Indeed, from \eqref{eq:barSigma} it is seen that $\bar  a=\frac{1}{\sigma}\left( \overline\Sigma+\sigma a\rho\right)\geq -k\sigma$, so Proposition \ref{prop:nirvana} is verified for $\gamma>1$ because there cannot be a complex root.

For $\gamma\in(0,1)$ there are some interesting cases of investor-nirvana occurrence:

\begin{example}[Infinite Information Premium]
Suppose $-\frac{1}{ 2}<\rho<0$, $-2\rho<\kappa<1$, $\sigma\geq 1$, and $ a$ such that $-2\rho\kappa\sigma< a<\sqrt{-2\kappa\rho\sigma}$. Then from \eqref{eq:partialInfoStability} it is seen that the partially-informed investor will never achieve nirvana for $\gamma\in(0,1)$ because
\[2\kappa+\frac{\overline\Sigma+\sigma a\rho}{\sigma^2}<0\ ,\]
but from \eqref{eq:expStability} it is seen that the fully-informed investor will achieve nirvana as $\gamma$ tends toward zero because
\[2\kappa\rho+\frac{ a}{\sigma}>0\ .\]
Hence, the information premium is infinite if the investors are given enough time.
\end{example}

\begin{example}[Undefined Information Premium]
The parameters can be selected so that the information premium from Section \ref{sec:premium} is undefined (i.e., equal to the difference $\infty-\infty$). Suppose $-\frac{1}{\sqrt 2}<\rho<0$ and $ a =-\frac{\kappa\sigma}{\rho}$. Then $\overline\Sigma =  a\sigma\sqrt{1-\rho^2}$, and the partially-informed investor will achieve nirvana as $\gamma$ tends toward zero because \eqref{eq:partialInfoStability} is violated,
\[2\kappa+\frac{\overline\Sigma+\sigma a\rho}{\sigma^2}=2\kappa+\frac{ a(\sqrt{1-\rho^2}+\rho)}{\sigma}>0\ .\]
The fully-informed investor will also achieve nirvana because \eqref{eq:expStability} is violated
\[2\kappa\rho+\frac{ a}{\sigma}>0\ .\]
Hence, the information premium is undefined if both investors have a long enough investment period.
\end{example}

\begin{example}[Simulation of Paths]
\label{ex:linearExample}
For $\gamma>1$ the Riccati equation for $A(t)$ can be solved explicitly, which allows for easy simulation of the BSDE solutions and the $G$ functions under both partial and full information. For $\gamma>1$ it follows that $A_+>0>A_-$, so $A_-$ is the long-term equilibrium of $A(t)$, and equations \eqref{eq:aPrime} and \eqref{eq:H_Primes} have explicit solutions,
\begin{align*}
A(t) &= A_-\frac{1-e^{-D(T-t)}}{1-\frac{A_-}{A_+}e^{-D(T-t)}}\\
H(t)&=a^2A_-\left((T-t)-\frac{2}{cA_-}\log\left(\frac{A_+-A_-e^{-D(T-t)}}{A_+-A_-}\right)\right)\ ,
\end{align*}
where $A_\pm$ is given by \eqref{eq:aPlusMinus} and $D =2\sqrt{\left(\kappa-\frac{(1-\gamma)\rho a }{\gamma\sigma}\right)^2-\frac{(1-\gamma) a^2}{\gamma\sigma^2}\left(1+\frac{(1-\gamma)\rho^2}{\gamma}\right)}$, and where $c$ is the same as that used in \eqref{eq:rootRiccati}. As $\gamma>1$ it follows that $D>0$, and so the solution is stable for large $T$.

Figure \ref{fig:simulationsLinear} shows a simulation of the linear model with the parameters given in Table \ref{tab:parametersLinear}. The simulation is informative because it shows how paths of $G(t)$ and $G^\full$ compare; in particular it shows how it is possible for $G(t)<G^\full(t,Y(t))$ even though Proposition \ref{prop:infoPremium} has shown $G(t)\geq\E[G^\full(t,Y(t))|\Fpartial_t]$ for $\gamma>1$. 
\begin{table}[h!]
\centering
\begin{tabular}{cccccc}
\multicolumn{6}{c}{Parameter Values}\\
  $\kappa$& $a$&$\rho$&$\sigma$& $ T$& $\gamma$\\
\hline
 8& .3& -.8& .15&1& 1.2
\end{tabular}
\vspace{.2cm}
\caption{\small The parameters for the simulation shown in Figure \ref{fig:simulationsLinear}.}
\label{tab:parametersLinear}
\end{table}
\begin{figure}[htbp]
\centering
\begin{tabular}{cc}
	\includegraphics[width=2.8in]{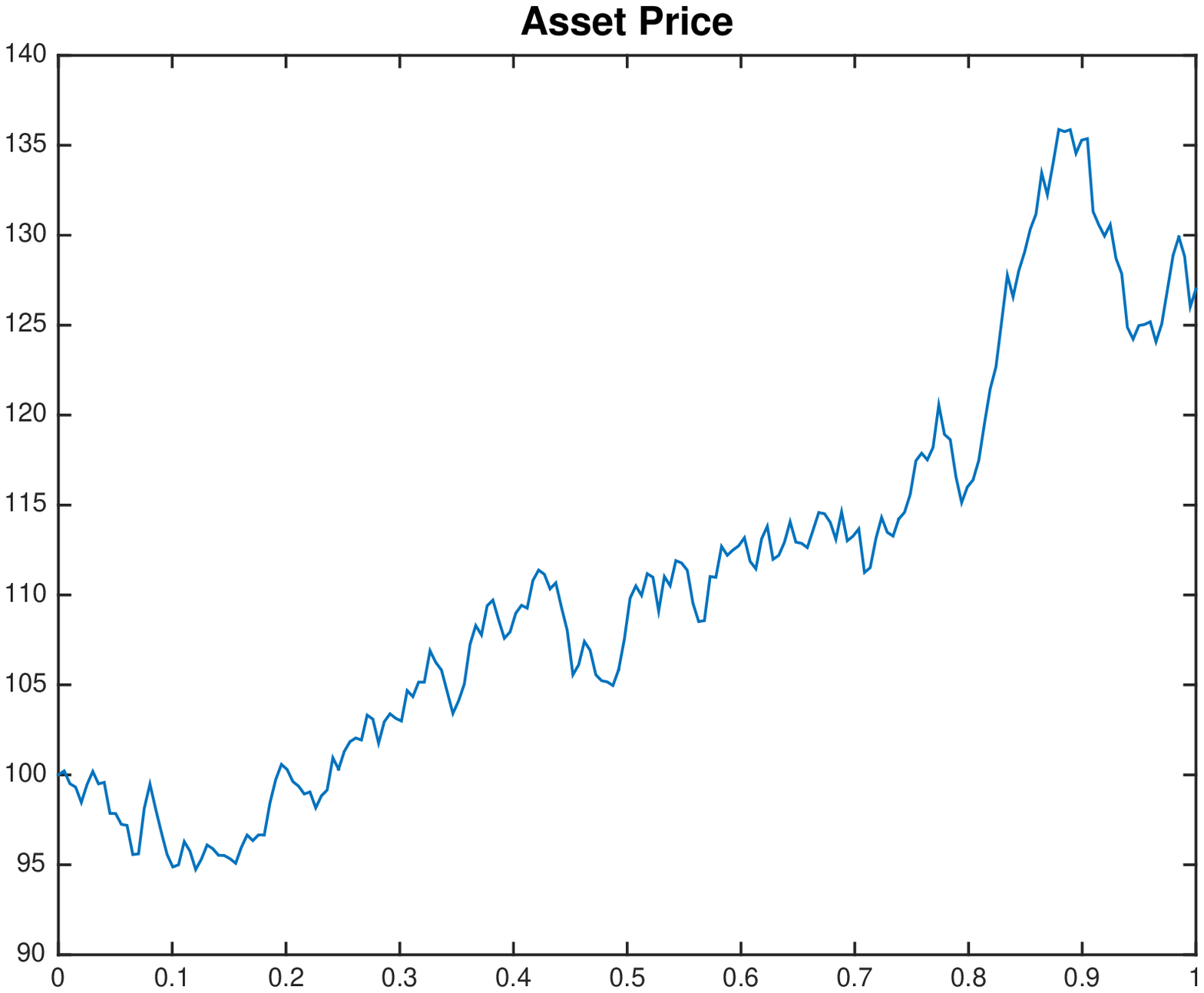}&
	\includegraphics[width=2.8in]{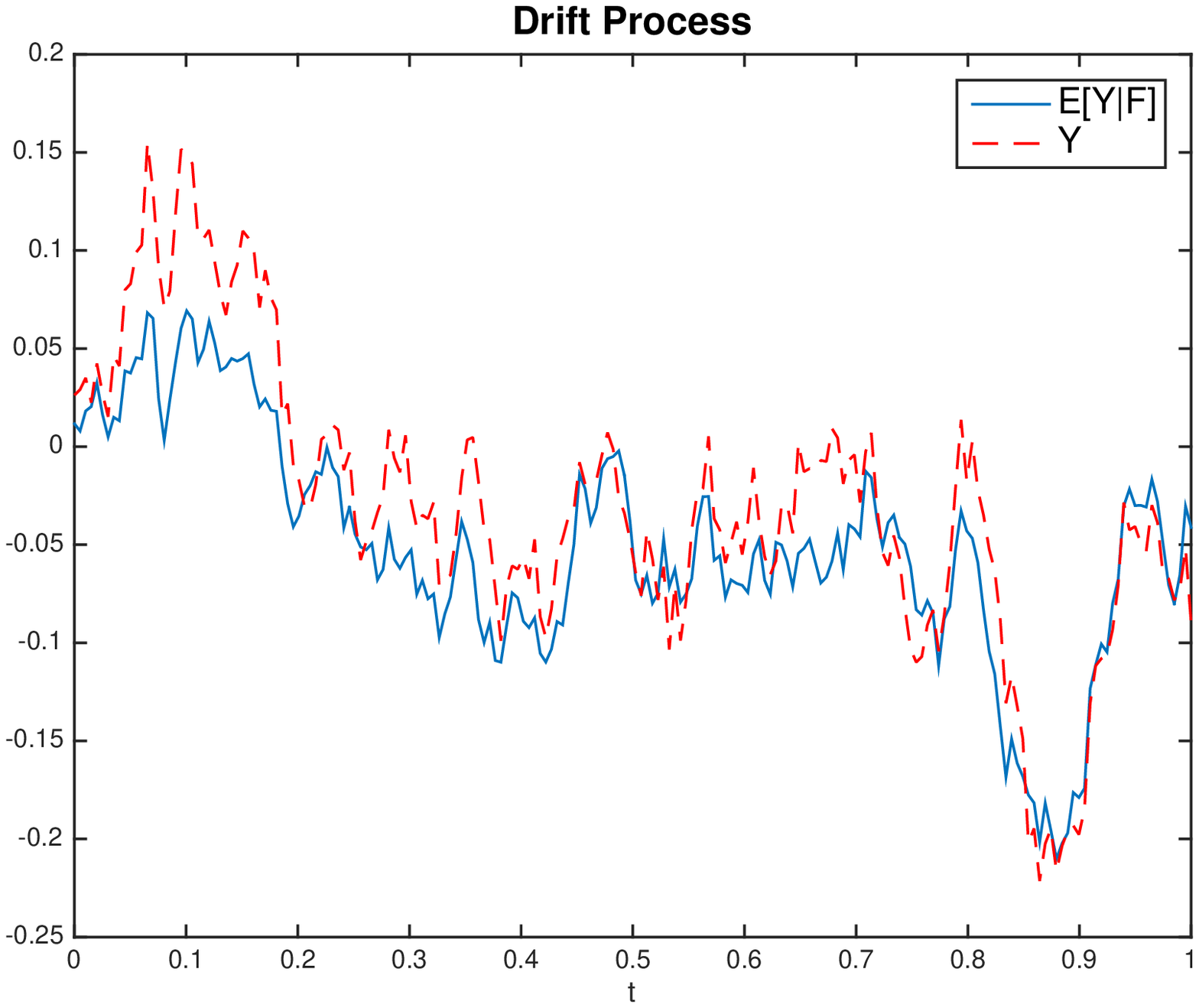}\\
	\includegraphics[width=2.8in]{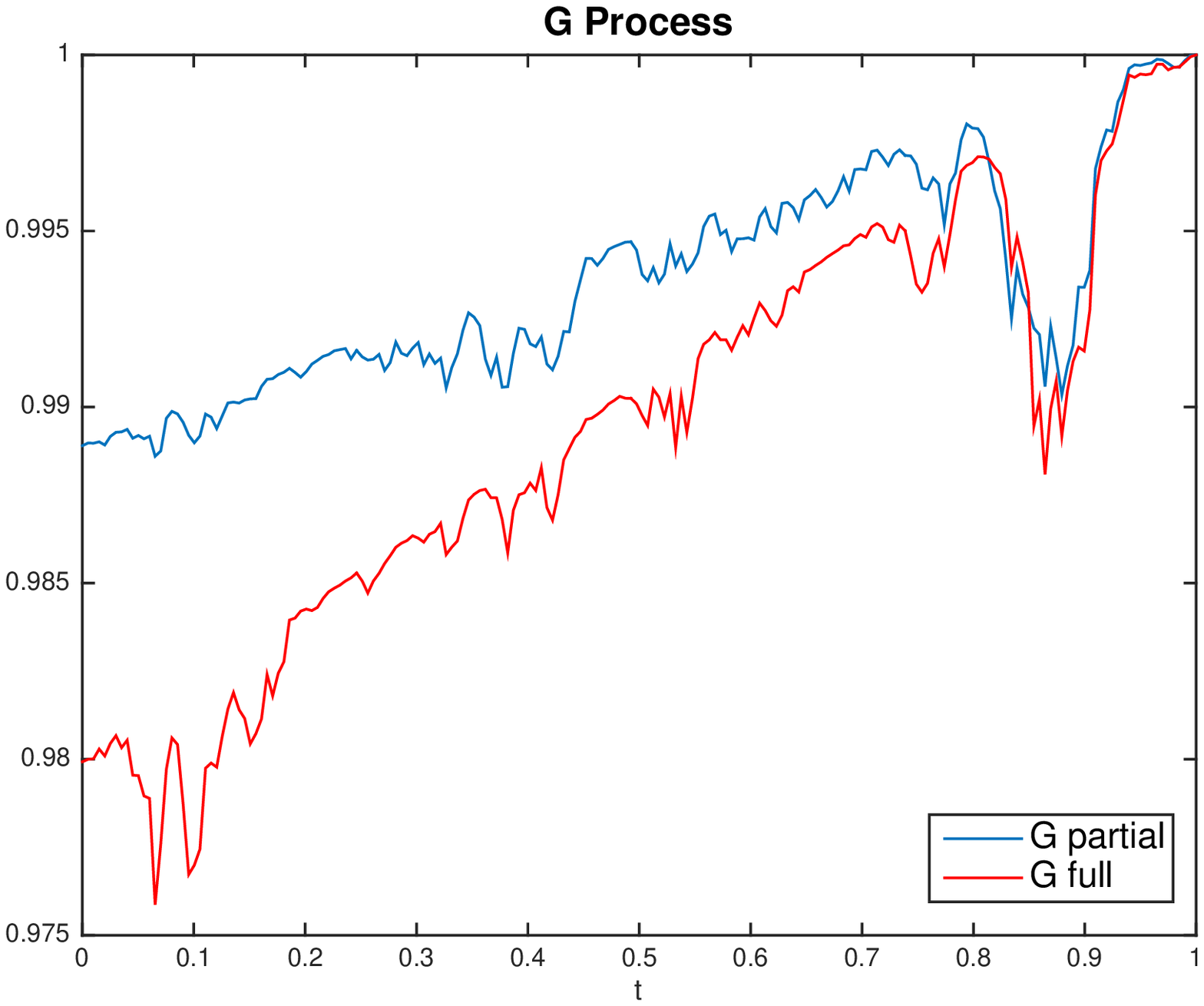}&
	\includegraphics[width=2.8in]{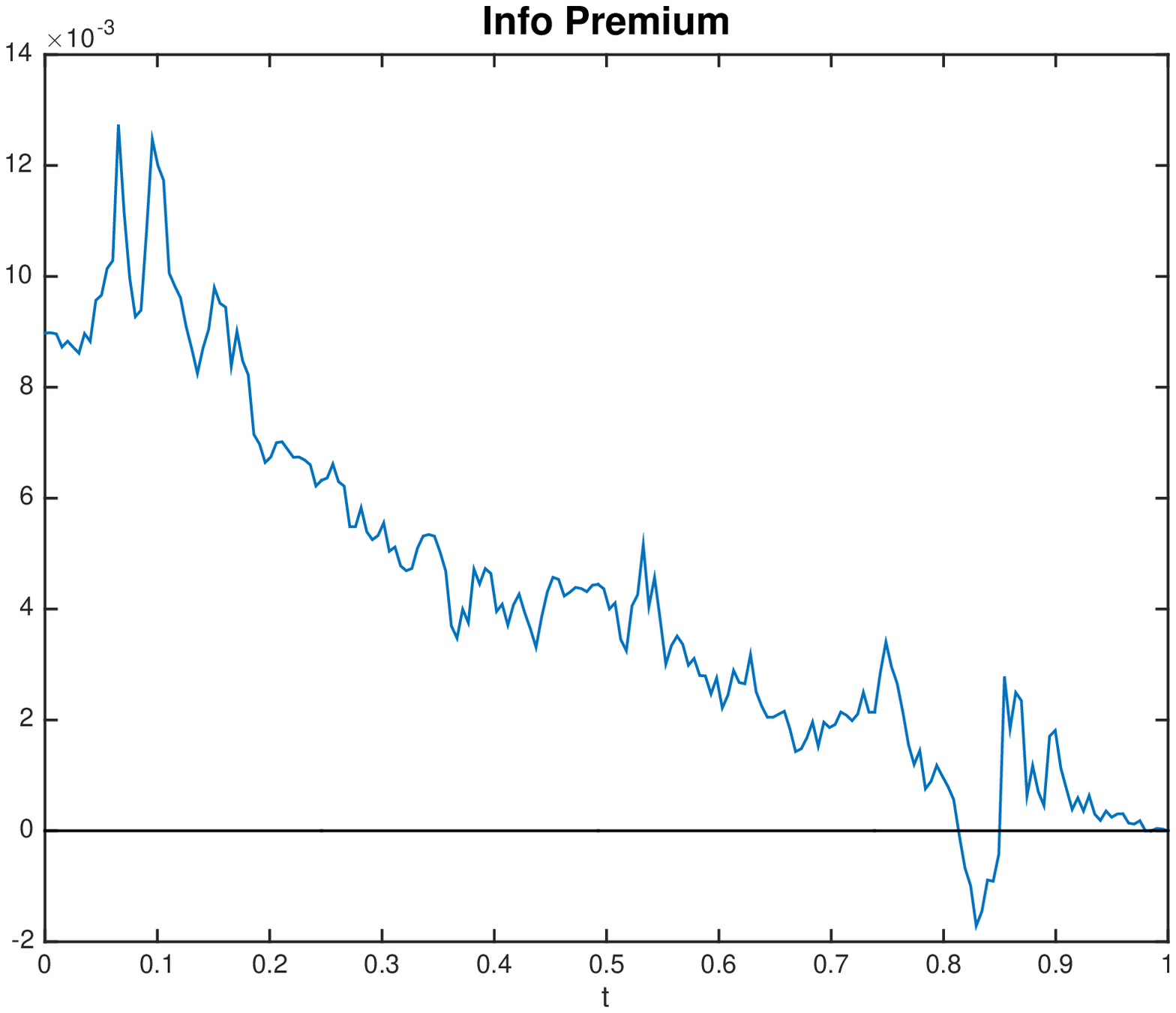}\\
\end{tabular}
\caption{\small Simulation of the linear model using the parameters of Table \ref{tab:parametersLinear}, with $\mu=r=0$. \textbf{Top Left:} The simulated asset price $S(t)$. \textbf{Top Right:} The simulated $Y(t)$ and the filter. \textbf{Bottom Left:} The BSDE solutions $G(t)$ and $G^\full(t,Y(t))$. \textbf{Bottom Right:} The difference $G(t)-G^\full(t,Y(t))$, for which there are a few times $t\in[0,T]$ when $G(t)<G^\full(t,Y(t))$ even though Proposition \ref{prop:infoPremium} has shown $G(t)\geq\E[G^\full(t,Y(t))|\Fpartial_t]$ for $\gamma>1$.}
\label{fig:simulationsLinear}
\end{figure}
\end{example}

\section{A Nonlinear Example}
\label{sec:nonlinearExample}

Recall the example from Remark \ref{remark:nonEmptyH}. Suppose that $Y(t)\in\mathbb R^1$ is a CIR process, and there is only one risky asset so that $S(t)\in\mathbb R^1$. The SDEs are
\begin{align}
\label{eq:nonlinearExample}
\frac{dS(t)}{S(t)} &= c\sqrt{Y(t)}dt+\sigma \left(\sqrt{1-\rho^2}dW(t)+\rho dB(t)\right)\\
\label{eq:nonlinearExample_dY}
dY(t)&=\kappa(\bar Y- Y(t))dt+ a\sqrt{Y(t)} dB(t)\ ,
\end{align}
with $0<a^2\leq 2\kappa\bar Y$, $\rho\in(-1,1)$, $c\in\mathbb R$, and $\bar Y>0$ being the long-term level of $Y(t)$. The wealth process is
\begin{align*}
\frac{d\Xpi(t)}{\Xpi(t)}&= rdt+\pi(t)\left(\frac{dS(t)}{S(t)}-rdt\right)\\
&=\Big(c\pi(t)\sqrt{Y(t)}+r(1-\pi(t))\Big)dt+\pi(t)\sigma \left(\sqrt{1-\rho^2}dW(t)+\rho dB(t)\right)\ . 
\end{align*}
In this example take $\gamma>1$ to avoid nirvana situations. For simplicity take $r=0$ and $\rho = 0$ so that the model is affine.

\subsection{Full Information}
The value function for power utility has an explicit solution. Similar to the fully-informed investor in the linear example of Section \ref{sec:nirvana}, it is shown in \cite{zariphopoulou2001} for ansatz 
\[V^\full(t,x,y) = U(x)G^\full(t,y)\ ,\]
that $G$ solves the PDE (in this case for $\rho =0$)
\begin{align*}
G_t^\full +\frac{ a^2y}{2}G_{yy}^\full+\kappa (\bar Y-y)G_y^\full+\frac{c^2(1-\gamma)}{2\gamma\sigma^2}yG^\full&=0\\
G^\full\Big|_{t=T}&=1\ .
\end{align*}
Using the ansatz,
\[G^\full(t,y) = \exp\Big(A(t)y+H(t)\big)\ ,\]
the solution uses functions $A$ and $H$ satisfying the equations
\begin{align}
\label{eq:aPrime2}
A'(t)+\frac{a^2}{2}A^2(t)-\kappa A(t)+\frac{c^2(1-\gamma)}{2\gamma\sigma^2}&=0\\
\label{eq:H_Primes2}
H'(t)+ \kappa\bar YA(t)&=0\ .
\end{align}
Similar to Example \ref{ex:linearExample}, equations \eqref{eq:aPrime2} and \eqref{eq:H_Primes2} have explicit solutions,
\begin{align*}
A(t) &= A_-\frac{1-e^{-D(T-t)}}{1-\frac{A_-}{A_+}e^{-D(T-t)}}\\
H(t)&=\kappa \bar YA_-\left((T-t)-\frac{2}{a^2A_-}\log\left(\frac{A_+-A_-e^{-D(T-t)}}{A_+-A_-}\right)\right)\ ,
\end{align*}
where 
\begin{align*}
A_\pm &= \frac{\kappa\pm\sqrt{\kappa^2-\frac{c^2(1-\gamma)}{\gamma\sigma^2}a^2}}{a^2}\\
D &=\sqrt{\kappa^2-\frac{c^2(1-\gamma)}{\gamma\sigma^2}a^2}\ .
\end{align*}

\subsection{Partial Information}

Direct simulation of $Z(t)$ from equation \eqref{eq:Z} allows for a numerical approximation of the first component of the solution to BSDE \eqref{eq:partialInfo_BSDE}. Namely, an approximation of $\xi$ from the dual value function in \eqref{eq:Vstar_rep} with a Monte Carlo expectation, where the expectation to be approximated is simplified using It\^o's lemma as done in the proof of Proposition \ref{prop:sup_xiK_bound},
\begin{align*}
\xi(t)&=\E\left[\left(\frac{Z(T)}{Z(t)}\right)^{-\frac{1-\gamma}{\gamma}}\Big|\Fpartial_t\right]\\
&=\E\left[\exp\left(\frac{(1-\gamma)c^2}{2\gamma^2\sigma^2}\int_t^T\widehat Y(u)du\right)\Big|\Fpartial_t\right]\\
&\approx \frac1N\sum_{\ell=1}^N\exp\left(\frac{(1-\gamma)c^2}{2\gamma^2\sigma^2}\int_t^T\widehat Y^{(\ell,t)}(u)du\right) \ ,
\end{align*}
for sample size $N$, where for each $\ell$ there is an independent sample $(\widehat Y^{(\ell,t)}(u))_{u\in[t,T]}$ conditional on $\Fpartial_t$. Samples of $\widehat Y^{(\ell,t)}(t)$ are obtained from forward sequential Monte Carlo (SMC) and computation of the filter. To compute the filter, one can either compute a particle filter for each trajectory of $S$, or one can approximate $Y$ with a finite-state Markov chain and then use a repeated application of Bayes rule over a small time step. The latter approach has been taken here because it is both fast and accurate (i.e., because $Y$ does not have a heavy tail) for this model. Note that simulation of $\xi(t)$ is like a branching process: for two times $t,t+\Delta t\in[0,T]$ the particles initialized at time $t$ cannot be reused for the simulation of particles to be initialized at $t+\Delta t$ (see \cite{touzi2014} for more on branching processes' relation to BSDEs).

The optimal value function is
\[V(t,x)= U(x) \xi(t)^\gamma\ ,\]
and so the information premium is seen by comparing $G(t)=\xi(t)^\gamma$ to $G^\full(t,Y(t))$. Using Jensen's inequality in the same manner as in Remark \ref{remark:nonEmptyH}, Condition \ref{cond:novikov} (Novikov) is satisfied if
\[\frac{c^2T}{2\sigma^2}<\frac{2\kappa}{a^2}\ ,\]
in which case $Z(t)$ is a true $\Fpartial_t$ martingale.

Figure \ref{fig:simulations} shows a comparison of full and partial information for realizations obtained with parameters from Table \ref{tab:parameters}. Noteworthy aspects in this example are: 
\begin{itemize}
\item Compared to the filters in Figure \ref{fig:simulationsLinear}, the filters in Figure \ref{fig:simulations} do not do as good of a job tracking the hidden drift $\sqrt{Y(t)}$. The reason is because the linear example has a strong correlation of $-.8$, which increases the \textit{signal-to-noise ratio (SNR)}. In contrast, this nonlinear example has zero correlation and hence much lower SNR.
\item Compared to the coefficients $G$ shown in Figure \ref{fig:simulationsLinear}, the partial-information $G$ in Figure \ref{fig:simulations} is smoother. This is due to the lack of tracking in the filer (see previous bullet point). 
\item The coefficients $G$ in Figure \ref{fig:simulations} have steeper slopes than those in Figure \ref{fig:simulationsLinear}. This is because the filter $\widehat Y(t)$ is almost constant in time, $\widehat Y(t) \approx \bar Y= .05$, which means positive average portfolio return, and $G(t)\approx \exp\left(\frac{(1-\gamma)\bar Y_t^2}{2\gamma\sigma^2}\right)$. Comparatively, the linear example has parameters chosen so that $\widehat Y(t)\approx 0$ for a net-zero average return. In other words, the parameters are such that the Sharpe ratios are higher in this nonlinear example.
\item The optimal $\pi$ for partial information has not been computed because no numerical method was proposed. 
\end{itemize}
This fourth point is a reiteration of a comment in Remark \ref{remark:partialInfoBSDE}, where it was pointed out that $\theta$ from the martingale representation is difficult to compute and requires a special numerical method; a numerical method for $\alpha$ would accomplish as much. In general, these bullets points highlight possible topics for future exploration in the area of numerical BSDE.

Finally, it should be pointed out that the information premium is low in this nonlinear example, which is seen by observing that $G$ and $G^\full$ are close together in Figure \ref{fig:simulations}. The reason for this is because $\bar Y > r= 0$ with Sharpe ratios $\hat Y_t /\sigma\approx\bar Y/\sigma = .33$ and equal to $1.92$ for $\sigma = .15$ and $\sigma=.026$, respectively, and so both the partially and fully-informed investors are placing a significant portion of their wealth into the risky asset. Comparatively, the linear case of Example \ref{ex:linearExample} would have a more pronounced premium if $\rho=0$; this would be the case because of low SNR, in which case the filter remains close to zero (i.e., $\hat Y_t \approx 0$ for all $t$) causing the Sharpe ratio to be very close zero, and therefore the partially-informed investor would invest very little in the risky asset and experience none of the improved portfolio returns.

\begin{figure}[htbp]
\centering
\begin{tabular}{cc}
\includegraphics[width=2.8in]{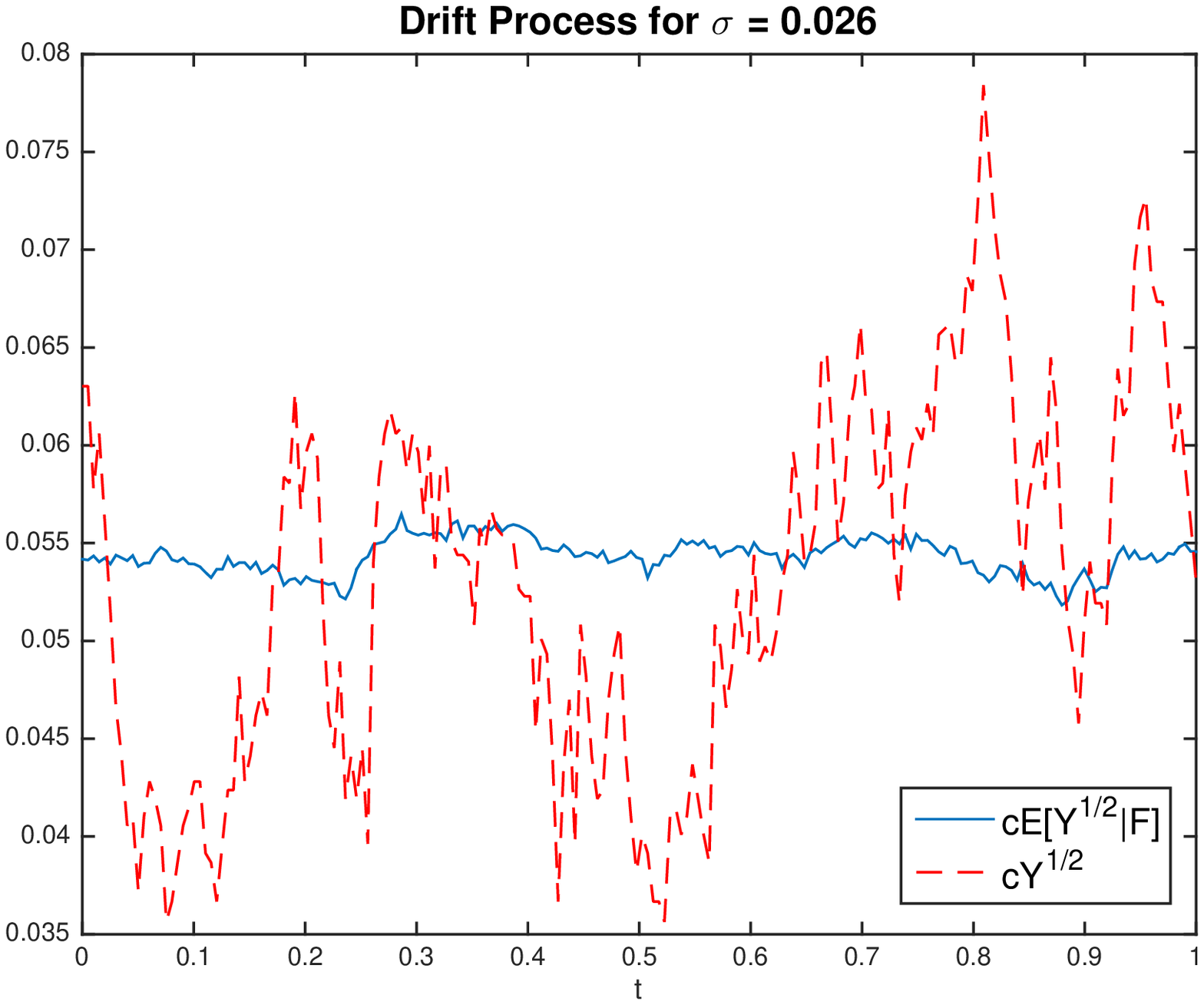}&
	\includegraphics[width=2.8in]{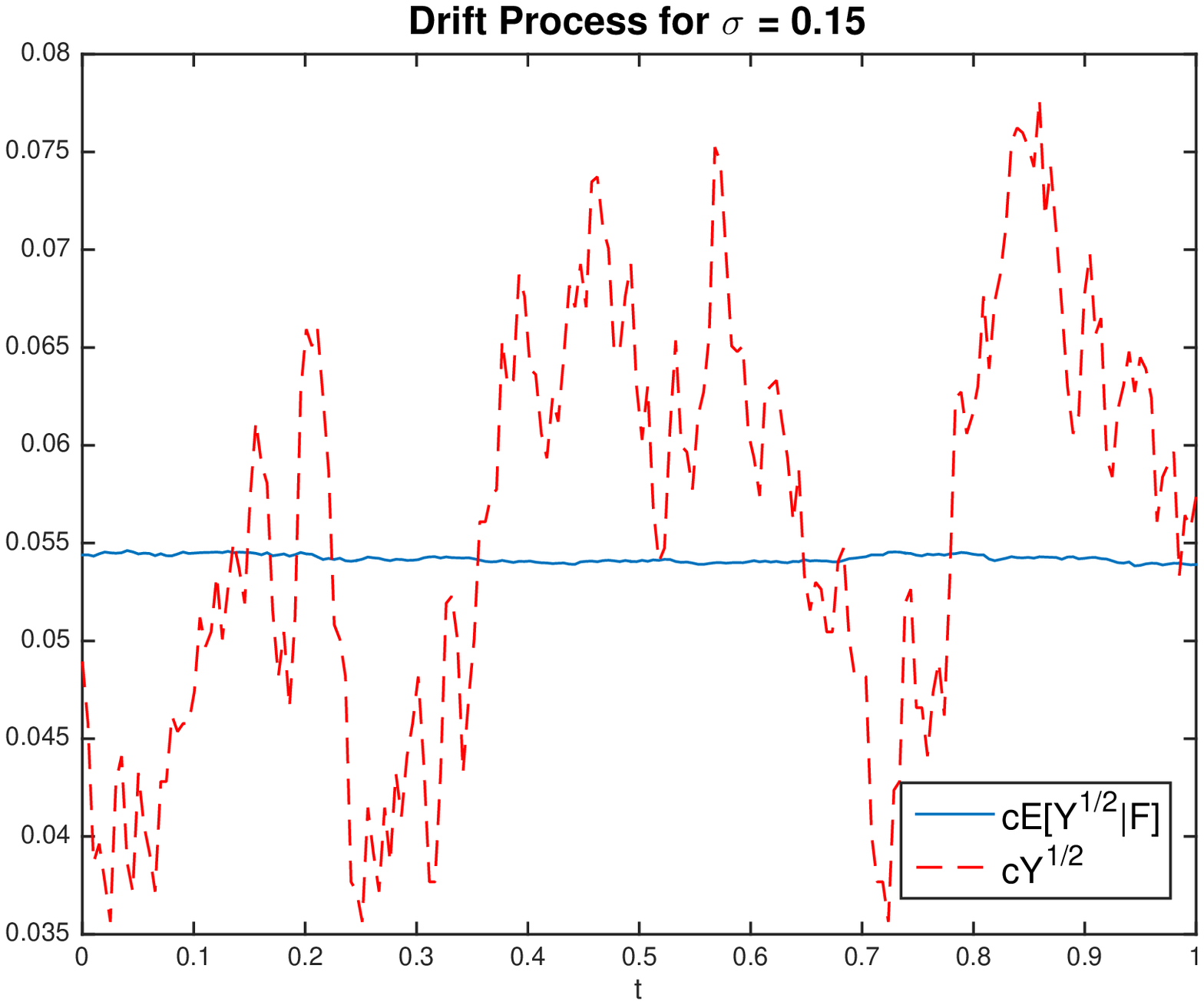}\\
	\includegraphics[width=2.8in]{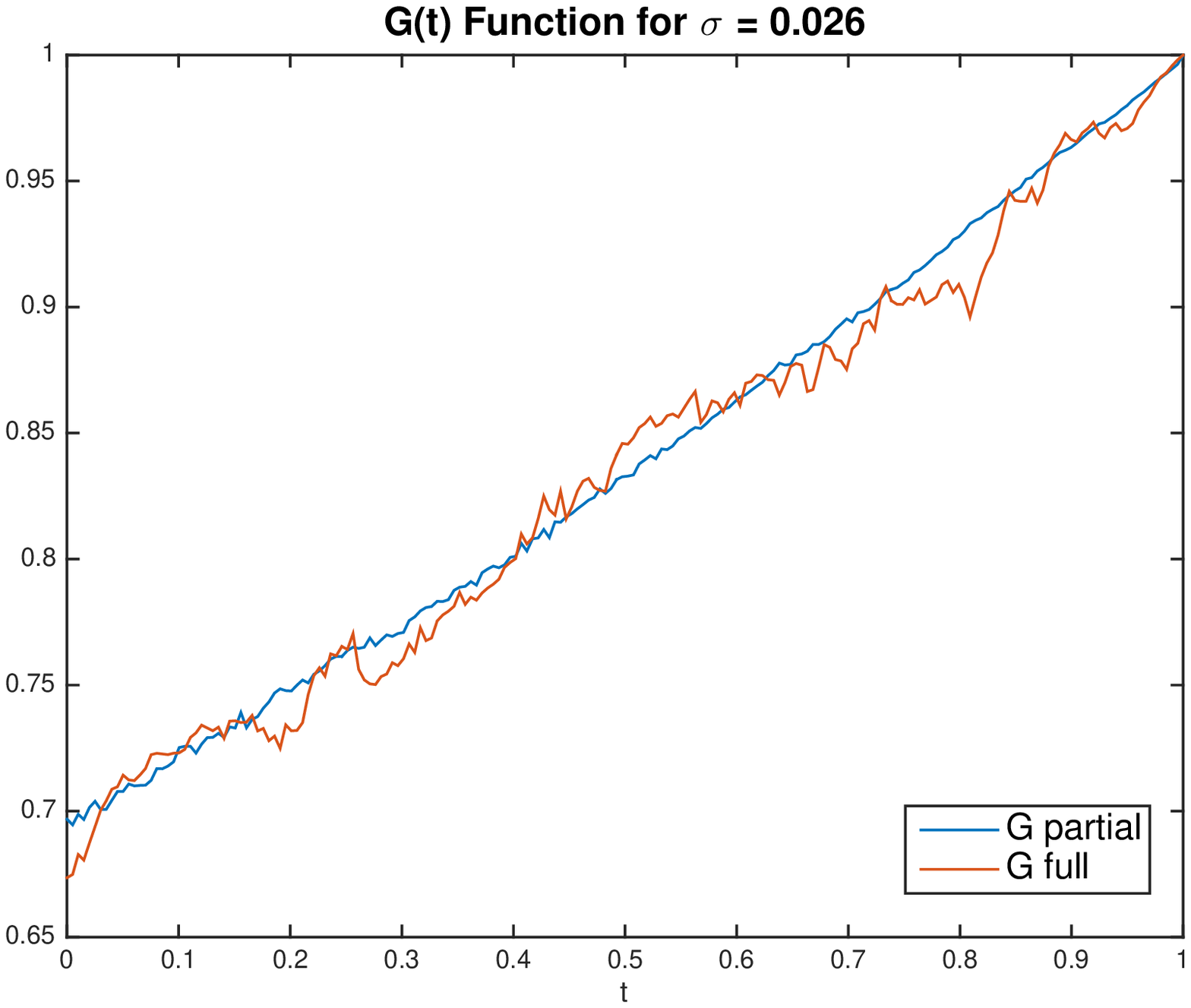}&
	\includegraphics[width=2.8in]{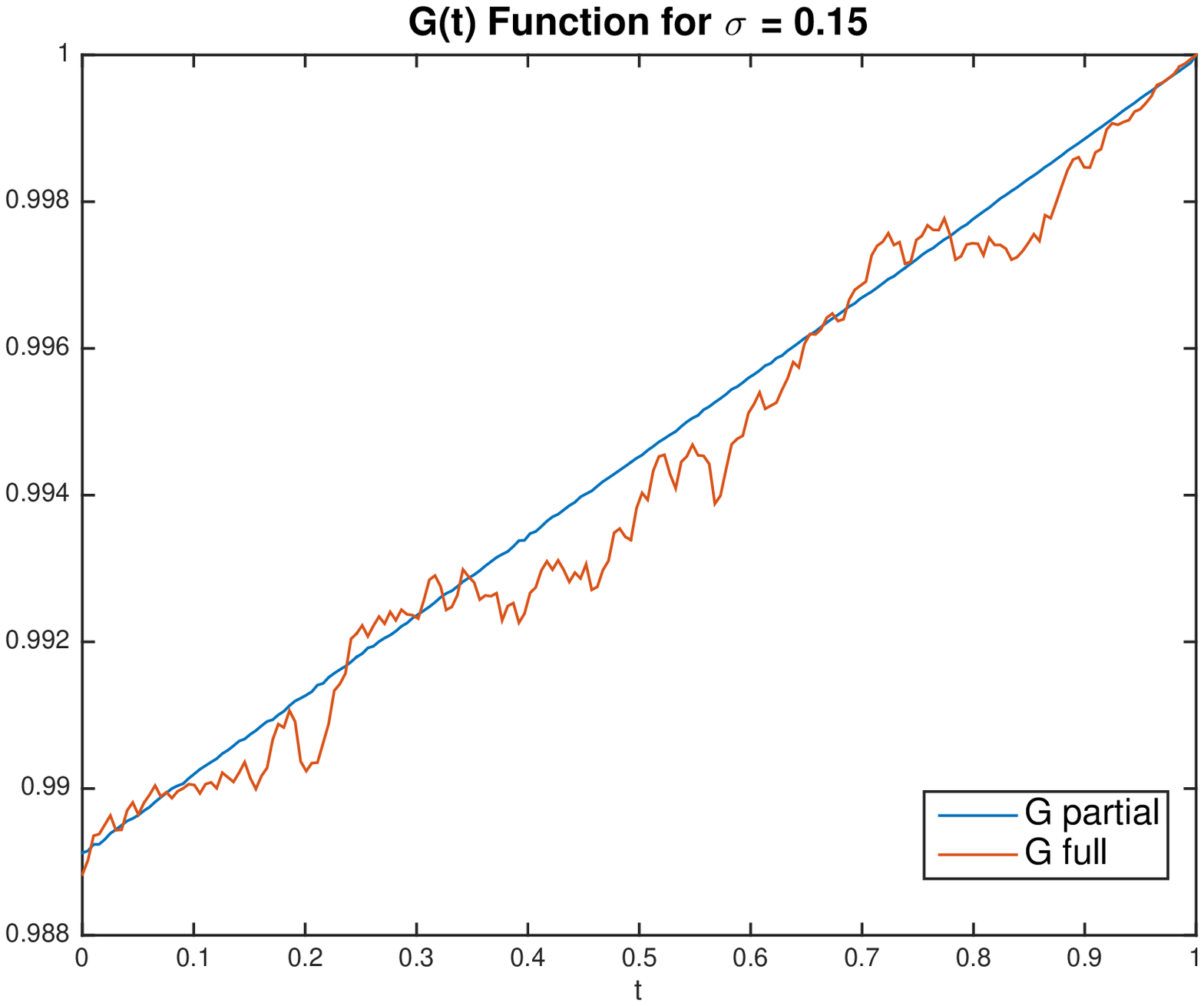}
\end{tabular}
\caption{\small A low-noise simulation with $\sigma=.026$ and a high-noise example with $\sigma=.15$. \textbf{Top Left:} The simulated low-noise $Y(t)$ and its filter. \textbf{Top Right:} The simulated high-noise $Y(t)$ and its filter. \textbf{Bottom Left:} The low-noise $G(t)$'s with sample size $N=10$. \textbf{Bottom Right:} A high-noise $G(t)$'s with sample size $N= 10 $.}
\label{fig:simulations}
\end{figure}
\begin{table}[h!]
\centering
\begin{tabular}{cccccc}
\multicolumn{6}{c}{Parameter Values}\\
 $c$& $\kappa$& $\bar Y$& $a$& $ T$& $\gamma$\\
\hline
.25& 8& .05& .4& 1& 1.2
\end{tabular}
\vspace{.2cm}
\caption{\small Parameter values for the nonlinear example in equations \eqref{eq:nonlinearExample} and \eqref{eq:nonlinearExample_dY}. Different values of $\sigma$ are tested, namely a low value of $.026$ and a high value of $.15$. Note the if the value of $\sigma$ is too low then Condition \ref{eq:novikov} will fail and it is possible for $Z(t)$ to have $\mathbb E[Z(T)/Z(t)|\Fpartial_t]<1$.}
\label{tab:parameters}
\end{table}

\section{Summary \& Conclusions}
\label{sec:concludes}

Investment with filtering under partial information is a non-Markov control problem, but also has some simplicity because the model can be reduced to a complete market. For the case of investors with a power utility function, the dual value function is shown to be the solution to a BSDE. The optimal strategy is also shown to be expressed in terms of the solution to the BSDE, and can be broken into two components: a myopic component where point estimate of $Y_t$ is inserted into the standard Merton problem, and a hedging term due to stochastic drift. In comparison with full information, the information premium is defined to be the expected loss in utility (from the perspective of the partially informed investor), and quantified in terms of the coefficients of the BSDEs.

A possible direction for future work on this problem is on the development of numerical methods for solving the partial-information BSDE; the proposed Monte Carlo approximation of Section \ref{sec:nonlinearExample} is a small step towards this goal. Monte Carlo and particle filtering will be useful, but there is likely to be an exponentially growing number of states taken by the filter, and so further innovation is needed. 

\appendix
\section{Proofs for Section \ref{sec:partialInfo_BSDE}}
\label{app:partialInfo_BSDE}
\begin{proposition}
\label{prop:thetaIntegrability}
If Condition \ref{cond:mgf_Z} holds, then $\theta\in\Hspace_T^2(\Pspace_\dd)$ where $\theta$ is the martingale representation in \eqref{eq:M-martingaleRep}.
\end{proposition}
\begin{proof}
A stochastic integral is a local martingale, so there is an increasing family of stopping times $(\tau_j)_{j=1,2,3,\dots}$ such that $\tau_j\nearrow\infty $ almost surely and $\int_0^{t\wedge\tau_j}\theta(u)^\tr d\zeta(u)$ is a true martingale. Then
\begin{align*}
&\E \int_0^{T\wedge\tau_j}\|\theta(t)\|^2 dt\\
& =-2\E\left[ -\frac12\int_0^{T\wedge\tau_j}\|\theta(t)\|^2 dt+\int_0^{T\wedge\tau_j}\theta(t)^\tr d\zeta(t)\right]\\
&=-2\E\log\left( M(T\wedge\tau_j)\Big/M(0)\right)\qquad\qquad\hbox{(because $dM(t) = M(t)\theta(t)^\tr d\zeta(t)$ in \eqref{eq:M-martingaleRep})}\ ,\\
&=-2\E\log\left( \E\left[Z(T)^{-\frac{1-\gamma}{\gamma}}\Big|\Fpartial_{T\wedge\tau_j}\right]\Big/M(0)\right)\\
& \leq-2\E\log\left( Z(T)^{-\frac{1-\gamma}{\gamma}}\right)+2\log M(0)\qquad\qquad\qquad\hbox{(Jensen's inequality)}\\
&=\frac{2(1-\gamma)}{\gamma}\E\log Z(T)+2\log M(0)\\
&=-\frac{1-\gamma}{\gamma}\E\int_0^T\left\|\sigma^{-1}(\hat h(t)-\mathbf{r} )\right\|^2 dt+2\log M(0)\\
&<\infty\ .
\end{align*}
This implies $\E\int_0^T\|\theta(u)\|^2 du\leq\liminf_j\E\int_0^{T\wedge\tau_j}\|\theta(u)\|^2 du<\infty$.
\end{proof}
\begin{proposition}
\label{prop:xiAnd_xim_paths}
Let $(\xi,\alpha)\in\Sspace_T^2(\Pspace_1)\times \Hspace_T^2(\Pspace_\dd)$ be a solution to \eqref{eq:partialInfo_BSDE}, and let $(\xi_K,\alpha_K)\in\Sspace_T^2(\Pspace_1)\times \Hspace_T^2(\Pspace_\dd)$ be the unique solution in $\Sspace_T^2(\Pspace_1)\times \Hspace_T^2(\Pspace_\dd)$ for the bounded BSDE in \eqref{eq:partialInfo_boundedBSDE} (in fact $\xi_K\in\Sspace_T^\infty(\Pspace_1)$ as shown in \cite{kobylanski2000}). For the stopping time 
\[\tau_K=\inf\left\{t>0~~\hbox{s.t}~~\|\hat h(t)\|\geq K\right\}\ ,\]
the solutions are equal for all $\omega\in\Omega$ such that $\tau_K\geq T$. That is, $(\xi(t)-\xi_K(t))\indicator{\tau_K>T}=0$ for all $t\in[0,T]$, and $(\alpha(t)-\alpha_K(t))\indicator{\tau_K>T}=0$ for all $t\in[0,T]$.
\end{proposition}
\begin{proof} The proof is by contradiction. Letting $\mathcal O=\{\omega\in\Omega~~\hbox{s.t}~~\tau_K\geq T\}$. Suppose $(\xi,\alpha)\neq (\xi_K,\alpha_K)$ for some $\omega\in\mathcal O$. Then there is another solution to \eqref{eq:partialInfo_boundedBSDE},

\[(\tilde \xi_K,\tilde\alpha_K)=\left\{
\begin{array}{ll}(\xi,\alpha)&\hbox{for }\omega\in\mathcal O\\
(\xi_K,\alpha_K)&\hbox{for }\omega\notin\mathcal O\ ,
\end{array}\right.
\]
with $(\tilde \xi_K,\tilde\alpha_K)\neq( \xi_K,\alpha_K)$, but the solution to \eqref{eq:partialInfo_boundedBSDE} is unique. Hence there is a contradiction.
\end{proof}

\begin{proposition}
\label{prop:sup_xiK_bound}
Let $(\xi_K,\alpha_K)\in\Sspace_T^\infty(\Pspace_1)\times \Hspace_T^2(\Pspace_\dd)$ be the unique solution to the BSDE in \eqref{eq:partialInfo_boundedBSDE}. If Condition \ref{cond:mgf_Z} holds, then $\sup_{K>0}\E\sup_{t\in[0,T]}|\xi_K(t)|^2<\infty$.
\end{proposition}
\begin{proof}
Recall the notation $\hat h_K(t)$ and $Z_K(t)$ from the proof of Theorem \ref{thm:partialInfo_BSDE}. Applying It\^o's lemma to $Z_K(t)^{-2\frac{1-\gamma}{\gamma}}$ yields a forward SDE,
\begin{align*}
d\left(Z_K(t)^{-2\frac{1-\gamma}{\gamma}}\right)&=\frac{(\gamma-1)(\gamma-2)}{\gamma^2}Z_K(t)^{-2\frac{1-\gamma}{\gamma}}\|\sigma^{-1}(\hat h_K(t)-\mathbf{r})\|^2dt\\
&\hspace{2cm}-2\frac{\gamma-1}{\gamma}Z_K(t)^{-2\frac{1-\gamma}{\gamma}}(\hat h_K(t)-\mathbf{r})(\sigma^{-1})^\tr d\zeta(t)\ .
\end{align*}
This SDE is a true semi-martingale because $\hat h_K$ is bounded, and so using variation of constants (i.e., integrating factor) and taking expectations yields an upper bound,
\begin{align*}
\E Z_K(T)^{-2\frac{1-\gamma}{\gamma}}&=\E\exp\left(\frac{(\gamma-1)(\gamma-2)}{\gamma^2}\int_0^T\|\sigma^{-1}(\hat h_K(t)-\mathbf{r})\|^2dt\right)\\
&\leq\E\exp\left(\frac{2|\gamma-1||\gamma-2|}{\epsilon\gamma^2}\int_0^T\left(\|\hat h(t)\|^2+\|\mathbf{r}\|^2\right)dt\right)\ ,
\end{align*}
where $\epsilon>0$ is the constant from \eqref{eq:sigmaBound} that bounds $\sigma$. Now notice the solution to BSDE \eqref{eq:partialInfo_boundedBSDE} has the following martingale bound,
\begin{align*}
\xi_K(t)&=\E\left[\left(\frac{Z_K(T)}{Z_K(t)}\right)^{-\frac{1-\gamma}{\gamma}}\Big|\Fpartial_t\right]\\
&=\E\left[\exp\left(-\frac{\gamma-1}{2\gamma^2}\int_t^T\|\sigma^{-1}(\hat h_K(u)-\mathbf{r})\|^2du\right)\Big|\Fpartial_t\right]\\
&\leq\E\left[\exp\left(-\frac{\gamma-1}{2\gamma^2}\int_0^T\|\sigma^{-1}(\hat h_K(u)-\mathbf{r})\|^2du\right)\Big|\Fpartial_t\right]\\
&=\E\left[Z_K(T)^{-\frac{1-\gamma}{\gamma}}\Big|\Fpartial_t\right]\ ,
\end{align*}
for which the last quantity is a continuous martingale, with continuity because it has a martingale-type-representation like that in equation \eqref{eq:M-martingaleRep}. Hence, from the Doob maximal inequality it is seen that
\begin{align*}
\E\sup_{t\in[0,T]}|\xi_K(t)|^2&\leq \E\sup_{t\in[0,T]}\Big|\E\left[ Z_K(T)^{-\frac{1-\gamma}{\gamma}}\Big|\Fpartial_t\right]\Big|^2\\
&\leq  4\E Z_K(T)^{-2\frac{1-\gamma}{\gamma}}\\
&\leq 4\E\exp\left(\frac{2|\gamma-1||\gamma-2|}{\epsilon\gamma^2}\int_0^T\left(\|\hat h(u)\|^2+\|\mathbf{r}\|^2\right)du\right)\\
&<\infty\ ,
\end{align*}
where the second inequality is from Doob and where finiteness is given by Condition \ref{cond:mgf_Z}, and hence the supremum over $K$ is finite. 
\end{proof}
\section{Verification Lemma for Full Information}
\label{app:fullInfoVerification}
This Appendix contains the verification proof for Proposition \ref{prop:verification} from Section \ref{sec:fullInformation}. For any admissible $\pi\in\mathcal A^\full$ consider the stopped SDE for $U(\Xpi(t)e^{r(T-t)})\goodchi(t)$, and let $\tau_k$ be an increasing sequence of stopping times with $\tau_k\wedge T\rightarrow T$ a.s. and for which the stochastic integral is a true martingale. The expectation satisfies
\begin{align*}
&\E\left[U\left(\Xpi(T\wedge\tau_k)e^{r(T-T\wedge\tau_k)}\right)\goodchi(T\wedge\tau_k)\Big|X(0)=x,Y(0)=y\right]\\
&=U\left(xe^{rT}\right)\goodchi(0)\\
&+(1-\gamma)\E\left[\int_0^{T\wedge\tau_k}U\left(\Xpi(u)e^{r(T-u)}\right)\Bigg(\goodchi(u)\pi(u)^\tr\Big(h(Y(u))-\mathbf{r}-\frac{\gamma}{2}\Sigma\pi(u)\Big)\right.\\
&\hspace{3cm}+\left.\pi(u)^\top\sigma_\y\psi(u)-F(Y(u),\goodchi(u),\psi(u))\Bigg)du\Bigg|X(0)=x,Y(0)=y\right]\\
&+\E\left[\int_0^{T\wedge\tau_k}U\left(\Xpi(u)e^{r(T-u)}\right)\goodchi(u)\Bigg((1-\gamma)\pi(u)^\tr(\sigma_\w dW(u)+\sigma_\y dB(t))\right.\\
&\hspace{7.5cm}\left.+\frac{\psi(u)}{\goodchi(u)}dB(u)\Bigg)\Bigg|X(0)=x,Y(0)=y\right]
\end{align*}
\begin{align*}
&=U\left(xe^{rT}\right)\goodchi(0)\\
&+(1-\gamma)\E\left[\int_0^{T\wedge\tau_k}U\left(\Xpi(u)e^{r(T-u)}\right)\Bigg(\goodchi(u)\pi(u)^\tr\Big(h(Y(u))-\mathbf{r}-\frac{\gamma}{2}\Sigma\pi(u)\Big)\right.\\
&\hspace{3cm}+\left.\pi(u)^\top\sigma_\y\psi(u)-F(Y(u),\goodchi(u),\psi(u))\Bigg)du\Bigg|X(0)=x,Y(0)=y\right]\\
&\leq U\left(xe^{rT}\right)\goodchi(0)\ ,
\end{align*}
where the inequality becomes an equality by inserting $F$ from \eqref{eq:F} and $\pi(u)=\pi^*(u,Y(u),\goodchi(u),\psi(u))$ given by equation \eqref{eq:optimalPifull}. Hence, 
\begin{align*}
&U\left(xe^{rT}\right)\goodchi(0)\\
&=\E\left[U\left(X^{\pi^*}(T\wedge\tau_k)e^{r(T-T\wedge\tau_k)}\right)\goodchi(T\wedge\tau_k)\Big|X(0)=x,Y(0)=y\right]\\
&\leq\sup_{\pi\in\mathcal A^\full}\E\left[U\left(\Xpi(T)\right)\Big|X(0)=x,Y(0)=y\right]\\
&= V^\full(0,x,y)\ .
\end{align*}
Verification is to show inequality in the other direction for the limit.
\subsection{Case $0<\gamma<1$}
For $0<\gamma<1$, using Fatou's lemma in the limit as $k\rightarrow \infty$ yields 
\begin{align*}
&\E\left[U(\Xpi(T))\Big|X(0)=x,Y(0)=y\right]\\
&=\E\left[\liminf_{k}U\left(\Xpi(T\wedge\tau_{k})e^{r(T-T\wedge\tau_k)}\right)\goodchi(T\wedge\tau_{k})\Big|X(0)=x,Y(0)=y\right]\\
&\leq\liminf_{k}\E\left[U\left(\Xpi(T\wedge\tau_{k})e^{r(T-T\wedge\tau_k)}\right)\goodchi(T\wedge\tau_{k})\Big|X(0)=x,Y(0)=y\right]\\
&\leq U\left(xe^{rT}\right)\goodchi(0)\ .
\end{align*}
The above calculations can be repeated for any $t\in[0,T]$, and hence 
\[V^\full(t,x,y)=\sup_{\pi\in\mathcal A^\full}\E\left[U(\Xpi(T))\Big|X(t)=x,Y(t)=y\right]\leq 
U\left(xe^{r(T-t)}\right)\goodchi(t)\ ,\]
which completes the verification for $\gamma\in(0,1)$.
\subsection{Case $\gamma>1$}
In this case $U(x)<0$ so Fatou lemma does not apply directly. Let $\Xpi_*(T)=\inf_{0\leq t\leq T}\Xpi(t)$ and assume $\E U(\Xpi_*(T))>-\infty$. Then
\begin{align*}
0&\leq\E\left[U\left(\Xpi(T)\right)- U(\Xpi_*(T))\Big|X(0)=x,Y(0)=y\right]\\
&=\E\left[\liminf_{k}\left(U\left(\Xpi(T\wedge\tau_k)e^{r(T-T\wedge\tau_k)}\right)-U(\Xpi_*(T))\right)\goodchi(T\wedge\tau_k)\Big|X(0)=x,Y(0)=y\right]\\
&\leq\liminf_{k}\E\left[\left(U\left(\Xpi(T\wedge\tau_k)e^{r(T-T\wedge\tau_k)}\right)-U(\Xpi_*(T))\right)\goodchi(T\wedge\tau_k)\Big|X(0)=x,Y(0)=y\right]\\
&=U\left(xe^{rT}\right)\goodchi(0)+\liminf_{k}\E\left[-U(\Xpi_*(T))\goodchi(T\wedge\tau_k)\Big|X(0)=x,Y(0)=y\right]\\
&\leq U\left(xe^{rT}\right)\goodchi(0)-\E\left[U(\Xpi_*(T))\Big|X(0)=x,Y(0)=y\right]\ .
\end{align*}
Now $\E\left[U(\Xpi_*(T))\Big|X(0)=x,Y(0)=y\right]$ cancels from both sides and there is the bound
 \[\E\left[U\left(\Xpi(T)\right)\Big|X(0)=x,Y(0)=y\right]< U\left(xe^{rT}\right)\goodchi(0)\ .\]
If it cannot be shown that $\E\left[U(\Xpi_*(T))\Big|X(0)=x,Y(0)=y\right]<\infty$, then a truncation argument can be used to show the bound up to an arbitrarily small constant.

\section{Proof of Theorem \ref{thm:fullInfo_uniqueness}}
\label{app:fullInfo_BSDE}
General existence and uniqueness theory for BSDEs can be applied if the problem is truncated to have $Y(t)$ and $\pi(t)$ confined to compact sets. For some positive $K<\infty$ define the truncated set of admissible strategies
\begin{equation*}
 \mathcal A_{K}^{\full }= \mathcal A^\full\cap\left\{\pi:[0,T]\times\Omega\rightarrow \mathbb R^\dd~~\hbox{s.t.}~\sup_{t\in[0,T]}\|\pi(t)\|<K~\hbox{a.s.}\right\}\ .
\end{equation*}
Also define the stopping time
\[\tau_K = \inf\left\{t>0~~\hbox{s.t}~~ \|Y(t)\|\geq K\right\}\ ,\]
and consider the truncated BSDE:
\begin{align}
\nonumber
-d\goodchi_K(t) &= (1-\gamma)F_K\left(Y(t),\goodchi_K(t) , \psi_K(t)\right)dt-\psi_K(t)^\tr dB(t)\ ,\qquad\hbox{for } t\leq \tau_k\\
\label{eq:fullInfo_boundedBSDE}
\goodchi_K(T\wedge\tau_K)&=1\ ,
\end{align}
where $F_K(y,g,\eta) = \max_{\|\pi\|\leq K}f(y,\pi,g,\eta)$ and is well defined because $f$ given by \eqref{eq:f_concave} is a concave function of $\pi$. There is a uniform Lipschitz constant for $F_K$ for all $t\leq \tau_K$, and so \eqref{eq:fullInfo_boundedBSDE} has a unique solution $(\goodchi_K,\psi_K)\in\Sspace_T^2(\Pspace_1^\full)\times \Hspace_T^2(\Pspace_{\qq}^\full)$. The solution to the BSDE is associated with a viscosity solution, $\goodchi_K(t)=G_K^\full(t,Y(t))$ and $\psi_K(t)=a(Y(t))^\tr\nabla G_K^\full(t,Y(t))$, where $G_K^\full$ is a viscosity solution of the boundary value problem,
\begin{align}
\label{eq:HJB_GK}
\left(\frac{\partial}{\partial t}+\mathcal L\right)G_K^\full+(1-\gamma)F_K\left(y,G_K^\full,\sigma_\y\nabla G_K^\full\right)&=0\\
\nonumber
G_K^\full\Big|_{t=T}&=1\\
\nonumber
G_K^\full\Big|_{\|y\|=K}&=1 \ .
\end{align}
Equation \eqref{eq:HJB_GK} has a unique classical solution, as it meets the criterions for application of Theorem 4.1 from Chapter IV.4 of \cite{flemingSonerBook}. Moreover, as $\mathcal L$ is degenerate elliptic and the Hessian $\nabla\nabla^\top G_K^\full$ is not present in the nonlinearity of \eqref{eq:HJB_GK}, the unique solution to \eqref{eq:HJB_GK} is also a viscosity solution (see \cite{CL1992}). Hence $\goodchi_K(t)=G_K^\full(t,Y(t))$ is a viscosity solution, and is the value function
\begin{align*}
&\goodchi_K^\full(t)=1+(1-\gamma)\\
&\times\sup_{\pi\in\mathcal A_{K}^\full}\E\left[\int_{t\wedge\tau_K}^{T\wedge\tau_K}f\Big(Y(u),\pi(u),G_K^\full(u,Y(u)),\sigma_\y\nabla G_K^\full(u,Y(u))\Big)du\Bigg|\F_t\right]\ .
\end{align*}
This truncated value function can be used to show uniqueness of solutions to \eqref{eq:fullInfo_BSDE}. The proof is based on the following two propositions,
\begin{proposition}
\label{prop:chiK_bound}
Suppose there exists $(\goodchi,\psi)\in\Sspace_T^2(\Pspace_1^\full)\times \Hspace_T^2(\Pspace_{\qq}^\full)$ a solution to \eqref{eq:fullInfo_BSDE}, in particular that $\E\sup_{t\in[0,T]}|\goodchi(t)|^2<\infty$. Then
\[\sup_{K>0}\E\sup_{t\in[0,T]}|\goodchi_K(t)|^2\leq \E\sup_{t\in[0,T]}|\goodchi(t)|^2<\infty\ ,\]
where $(\goodchi_K,\psi_K)\in\Sspace_T^2(\Pspace_1^\full)\times \Hspace_T^2(\Pspace_{\qq}^\full)$ is a solution to \eqref{eq:fullInfo_boundedBSDE}.
\end{proposition}
\begin{proof}
Start with the case $\gamma\in(0,1)$. For any $(t,y,g,p)\in[0,T]\times \mathbb R^\qq\times\mathbb R^+\times\mathbb R^\qq$, $F_K(y,g,p)\leq F(y,g,p)$. Hence, $0\leq\goodchi_K(t)=\goodchi_K(t\wedge\tau_K)\leq \goodchi(t\wedge\tau_K)$ by a comparison principle (see Proposition 2.9 in \cite{kobylanski2000}), and 
\[\sup_{K>0}\E\sup_{t\in[0,T]}|\goodchi_K(t)|^2\leq \sup_{K>0}\E\sup_{t\in[0,T]}|\goodchi(t\wedge\tau_K)|^2\leq \E\sup_{t\in[0,T]}|\goodchi(t)|^2<\infty\ ,\]
because $\sup_{t\in[0,T]}|\goodchi(t\wedge\tau_K)|^2\leq \sup_{t\in[0,T]}|\goodchi(t)|^2$. 

For $\gamma>1$ the comparison is made by looking at $0\geq(1-\gamma)F_K(y,g,p)\geq (1-\gamma)F(y,g,p)$, which implies $1\geq\goodchi_K(t)=\goodchi_K(t\wedge\tau_K)\geq \goodchi(t\wedge\tau_K)$. Taking expectations of squares yields $\sup_{K>0}\E\sup_{t\in[0,T]}|\goodchi_K(t)|^2\leq \E\sup_{t\in[0,T]}|\goodchi(t\wedge\tau_K)|^2\vee 1\leq \E\sup_{t\in[0,T]}|\goodchi(t)|^2\vee 1<\infty$.
\end{proof}

\begin{proposition}
\label{prop:chiK_inside}
Suppose there exists $(\goodchi,\psi)\in\Sspace_T^2(\Pspace_1^\full)\times \Hspace_T^2(\Pspace_{\qq}^\full)$ a solution to \eqref{eq:fullInfo_BSDE}. Then
\[(\goodchi(t)-\goodchi_K(t))\indicator{\tau_K\geq T}=0\qquad\hbox{almost surely for all $t\in[0,T]$,}\]
where $(\goodchi_K,\psi_K)\in\Sspace_T^2(\Pspace_1^\full)\times \Hspace_T^2(\Pspace_{\qq}^\full)$ is a solution to \eqref{eq:fullInfo_boundedBSDE}.
\end{proposition}
\begin{proof}The proof is by contradiction and (similar to that of Proposition \ref{prop:xiAnd_xim_paths}). Letting $\mathcal O=\{\omega\in\Omega~~\hbox{s.t}~~\tau_K\geq T\}$. Suppose $(\goodchi,\psi)\neq (\goodchi_K,\psi_K)$ for some $\omega\in\mathcal O$. Then there is another solution to \eqref{eq:partialInfo_boundedBSDE},

\[(\tilde \goodchi_K,\tilde\psi_K)=\left\{
\begin{array}{ll}(\goodchi,\psi)&\hbox{for }\omega\in\mathcal O\\
(\goodchi_K,\psi_K)&\hbox{for }\omega\notin\mathcal O\ ,
\end{array}\right.
\]
with $(\tilde \goodchi_K,\tilde\psi_K)\neq( \goodchi_K,\psi_K)$, but the solution to \eqref{eq:fullInfo_boundedBSDE} is unique. Hence there is a contradiction.
\end{proof}
Using the truncated problem and associated notation, Propositions \ref{prop:chiK_bound} and \ref{prop:chiK_inside} are applied to prove Theorem \ref{thm:fullInfo_uniqueness}:

\begin{proof}[Proof of Theorem \ref{thm:fullInfo_uniqueness}]
Let $(\goodchi,\psi)$ and $(\tilde\goodchi,\tilde\psi)$ be two solutions to \eqref{eq:fullInfo_BSDE} in the space $\Sspace_T^2(\Pspace_1^\full)\times \Hspace_T^2(\Pspace_{\qq}^\full)$. Applying Propositions \ref{prop:chiK_bound} and \ref{prop:chiK_inside} and taking the limit as $K\rightarrow\infty$,
\begin{align*}
\E\sup_{t\in[0,T]}|\goodchi(t)-\tilde\goodchi(t)|&\leq \E\sup_{t\in[0,T]}|\goodchi(t)-\goodchi_K(t)|+\E\sup_{t\in[0,T]}|\tilde\goodchi(t)-\goodchi_K(t)|\\
&=\E\sup_{t\in[0,T]}|\goodchi(t)-\goodchi_K(t)|\indicator{\tau_K<T}+\E\sup_{t\in[0,T]}|\tilde\goodchi(t)-\goodchi_K(t)|\indicator{\tau_K<T}\\
&\leq \left(2\E\sup_{t\in[0,T]}|\goodchi(t)|^2\E\indicator{\tau_K<T}\right)^{1/2}+\left(2\E\sup_{t\in[0,T]}|\tilde\goodchi(t)|^2\E\indicator{\tau_K<T}\right)^{1/2}\\
&\rightarrow 0\ ,
\end{align*}
and so $\goodchi=\tilde\goodchi$ almost surely.

To show uniqueness of $\psi$, consider the integrated form of the difference,
\begin{align*}
&\goodchi(t)-\tilde\goodchi(t)\\
&=\goodchi(t_0)-\tilde\goodchi(t_0)+(1-\gamma)\int_{t_0}^t \left(F(Y(u),\goodchi(u) , \psi(u)\right)-F\left(Y(u),\tilde\goodchi(u) , \tilde\psi(u))\right)du\\
&\hspace{7cm}-\int_{t_0}^t\left(\psi(u)-\tilde\psi(u)\right)^\tr dB(u)\ ,
\end{align*}
for all $0\leq t_0\leq t\leq T$. As $\goodchi=\tilde\goodchi$ almost surely, it follows that 
\begin{align}
\label{eq:diffFexpectation}
\E\left[\int_{t_0}^t \left(F\left(Y(u),\goodchi(u) , \psi(u)\right)-F(Y(u),\tilde\goodchi(u) , \tilde\psi(u))\right)du\Big|\mathcal F_{t_0}^B\right]&=0\ ,
\end{align}
and
\begin{align}
\nonumber
\E\Bigg[\left((1-\gamma)\int_{t_0}^t \left(F\left(Y(u),\goodchi(u) , \psi(u)\right)-F(Y(u),\tilde\goodchi(u) , \tilde\psi(u))\right)du\right.\hspace{2cm}&\\
\label{eq:diffXi2expectation}
\hspace{5cm}\left.-\int_{t_0}^t\left(\psi(u)-\tilde\psi(u)\right)^\tr dB(u)\right)^2\Bigg|\mathcal F_{t_0}^B\Bigg]&=0\ .
\end{align}
If the square inside the expectation of \eqref{eq:diffXi2expectation} is multiplied out, it is found by applying \eqref{eq:diffFexpectation} that the cross term is zero, in particular
\begin{align*}
&\E \int_0^T\int_0^T\left(F\left(Y(u'),\goodchi(u') , \psi(u')\right)\right.\\
&\hspace{3cm}\left.-F(Y(u'),\tilde\goodchi(u') , \tilde\psi(u'))\right)\left(\psi(u)-\tilde\psi(u)\right)^\tr dB(u)du'=0\ .
\end{align*}
Therefore, \eqref{eq:diffXi2expectation} is equal to the sum of two non-negative quantities, and because this sum is equal to zero, it follows that both quantities must be zero. Namely, 
\begin{align*}
\E\left(\int_0^T \left(F\left(Y(u),\goodchi(u) , \psi(u)\right)-F(Y(u),\tilde\goodchi(u) , \tilde\psi(u))\right)du\right)^2&=0
 ,\\
\E\left(\int_0^T\left(\psi(u)-\tilde\psi(u)\right)^\tr dB(u)\right)^2&=0 \ .
\end{align*} 
Hence, by the It\^o isometry
\[\E\int_0^T\|\psi(u)-\tilde\psi(u)\|^2du = \E\left(\int_0^T\left(\psi(u)-\tilde\psi(u)\right)^\tr dB(u)\right)^2=0\ ,\]
which means $\psi(t)=\tilde\psi(t)$ almost surely for almost everywhere $t\in[0,T]$.
\end{proof}

\bibliographystyle{alpha}
{\small\bibliography{refs}}

\end{document}